\documentclass[runningheads]{llncs}

\usepackage[utf8]{inputenc}
\usepackage[english]{babel}
\usepackage{csquotes}
\usepackage{stmaryrd}
\usepackage{xspace}
\usepackage{mathtools}
\usepackage[characters]{ltl} 
\usepackage{color}

\usepackage{graphicx}
\usepackage[firstpage]{draftwatermark}
\SetWatermarkText{\raisebox{13.9cm}{%
 \includegraphics{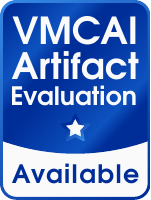} \hspace{7cm} \includegraphics{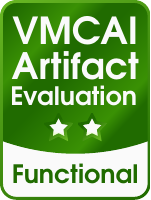}%
}}
\SetWatermarkAngle{0}

\usepackage{standalone}

\usepackage{listings,lstautogobble}
\usepackage{thmtools,thm-restate}

\usepackage{cleveref} % load last
\usepackage{multirow}
\usepackage{rotating} % Rotating table
\usepackage{adjustbox}
\usepackage{flushend}

\usepackage{xargs}
\usepackage{todonotes}
\newcommandx{\itodo}[2][1=]{\todo[inline,caption={},linecolor=red,backgroundcolor=red!25,bordercolor=red,#1]{#2}}

\newcommand{\ignore}[1]{}

\renewcommand{\phi}{\varphi}

\newcommand{\DLS}{\ensuremath{\mathrm{DLS}^*}}

\newcommand{\States}{\ensuremath{\mathsf{States}}}

\DeclarePairedDelimiter\abs{\lvert}{\rvert}%

\newcommand{\Const}{\mathcal{C}}
\newcommand\start{v_0}
\newcommand\term{\textsf{term}}

\newcommand{\Init}{{\sf Init}}
\newcommand\Transitions{\Delta}

\newcommand\traces{\ensuremath{\mathsf{Traces}}}

\newcommand{\rels}{\mathcal{R}}
\newcommand{\relsS}{{\ensuremath{\rels_{\mathit{state}}}}}
\newcommand{\relsSn}{{\ensuremath{\rels_{\mathit{state}}^n}}}
\newcommand{\relsI}{{\ensuremath{\rels_{\mathit{input}}}}}
\newcommand\relsA{{\ensuremath{\rels_{\mathcal A}}}}
\newcommand\relsB{{\ensuremath{\rels_{\mathcal B}}}}
\newcommand\pA{\ensuremath{\mathcal A}}
\newcommand\pB{\ensuremath{\mathcal B}}

\newcommand\wadd{\mathbin{{+}{=}}}
\newcommand\wsub{\mathbin{{-}{=}}}

\newcommand\sem[1]{\ensuremath{\llbracket #1 \rrbracket}}

\newcommand\fol{\textnormal{FOL}\xspace}

\newcommand\eufol{\text{$\exists^*\forall^*$\fol}\xspace}
\newcommand\uefol{\text{$\forall^*\exists^*$\fol}\xspace}

\renewcommand\implies{\rightarrow}

\newcommand\rank{\ensuremath{\mathit{qr}}}

\newcommand\Conf{\textit{Conflict}} %% conflict
\newcommand\Assign{\textit{Assign}} %% assignment
\newcommand\Read{\textit{Read}} %% reviews
\newcommand\Review{\textit{Review}} %% reviews

\newcommand\Next{\ensuremath{\mathit{next}}\xspace}

\newcommand\Leader{\ensuremath{\mathit{leader}}\xspace}
\newcommand\Msg{\ensuremath{\mathit{msg}}\xspace}
\newcommand\Btw{\ensuremath{\mathit{between}}\xspace}

\newcommand\err{\textit{error}}

\newcommand\ttrue{\ensuremath{\mathit{true}}\xspace}
\newcommand\tfalse{\ensuremath{\mathit{false}}\xspace}

\newcommand{\T}{\ensuremath{\mathcal{T}}}

\definecolor{darkgreen}{rgb}{0,0.6,0}

\lstdefinelanguage{workflows}{
  autogobble,
  columns=fullflexible,
  sensitive=true,
  commentstyle = \itshape, % \color{blue},
  morecomment={[l]//},
  mathescape=true,
  identifierstyle={\itshape},
  literate=
    {forall}{\textnormal{\textbf{forall}}\ }{2} 
    {may}{\textnormal{\textbf{may}}}{2} 
    {¬}{$\neg$}{2} 
    {!}{$\neg$}{2} 
    {<-}{$\leftarrow\ $}{2} 
    {->}{$\rightarrow\ $}{2} 
    {:=}{$\leftarrow\ $}{2} 
    {+=}{$\wadd\ $}{2} 
    {-=}{$\wsub\ $}{2}
}

\lstdefinelanguage{fots}{
  morekeywords={while, do, choose, or},
  autogobble,
  columns=fullflexible,
  sensitive=true,
  commentstyle = \itshape\color{darkgreen},
  morecomment={[l]//},
  mathescape=true,
  identifierstyle={\itshape},
  literate=
    {¬}{$\neg$}{2} 
    {∧}{$\land\ $}{2} 
    {∨}{$\lor\ $}{2} 
    {∃}{$\exists$}{2} 
    {∀}{$\forall$}{2} 
    {!}{$\neg$}{2} 
    {<-}{$\leftarrow\ $}{2} 
    {->}{$\rightarrow\ $}{2} 
    {:=}{$\coloneqq\ $}{2} 
    {⊥}{$\emptyset\ $}{2} 
    {+=}{$\wadd\ $}{2} 
    {-=}{$\wsub\ $}{2}
}

\usetikzlibrary{automata}
\tikzset{
  fotsnode/.style = {draw, circle},
  font=\footnotesize
}

\begin{document}

\title{
	How to Win First-Order Safety Games
    \thanks{
        \vtop{%
          \vskip-5pt
          \hbox{%
            \protect\includegraphics[width=1 cm]{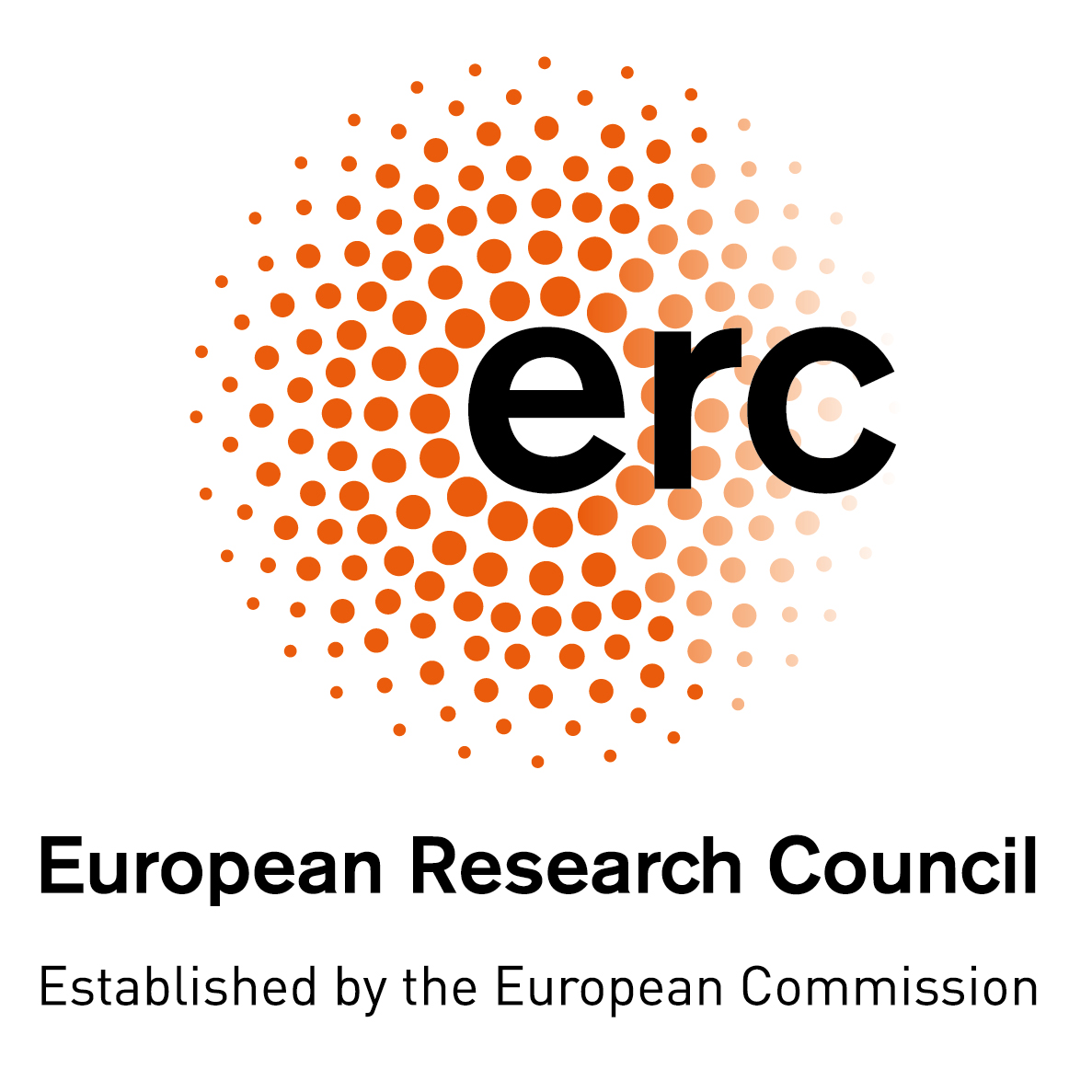}%
          }%
        }
        \parbox[t]{0.85\textwidth}{
        The project has received funding from the European Research Council (ERC) 
        under the European Union’s Horizon 2020 research and innovation programme 
        under grant agreement No. 787367 (PaVeS). 
        }
    }
} 

\author{
Helmut Seidl\inst{1}
% \orcidID{0000-0002-2135-1593} 
\and
Christian Müller\inst{1}
% \orcidID{0000-0001-9560-826X} 
\and
Bernd Finkbeiner\inst{2}
% \orcidID{?}
}
\authorrunning{H. Seidl et al.}
\institute{Technische Universität München\\
\email{\{seidl,christian.mueller\}@in.tum.de} \and
CISPA, Saarland University\\
\email{finkbeiner@cs.uni-saarland.de}}

\maketitle

\begin{abstract}
  \vspace{-1em}
  First-order (FO) transition systems have recently attracted attention for the verification 
of parametric systems such as network protocols, software-defined networks or 
multi-agent workflows like conference management systems.
Functional correctness or noninterference of these systems 
have conveniently been formulated as safety or hypersafety properties, respectively.
In this article, we take the step from verification to synthesis ---
tackling the question whether it is possible to automatically synthesize predicates
to enforce safety or hypersafety properties like noninterference.
For that, we generalize FO transition systems to FO safety games. 
For FO games with monadic predicates only, 
we provide a complete classification into decidable and undecidable cases.
For games with non-monadic predicates, we concentrate on universal first-order invariants,
since these are sufficient to express a large class of properties --- for example
noninterference.
We identify a non-trivial sub-class where invariants can be proven inductive and
FO winning strategies be effectively constructed. 
We also show how the extraction of weakest FO winning strategies can be reduced
to SO quantifier elimination itself.
We demonstrate the usefulness of our approach by 
automatically synthesizing nontrivial FO specifications of messages
in a leader election protocol as well as for 
paper assignment in a 
conference management system to exclude unappreciated disclosure of reports.

\end{abstract}


\section{Introduction}\label{s:intro}

Given a network of processes, can we synthesize the content of messages to be sent to
elect a single leader?
Given a conference management system, can we automatically synthesize a
strategy for paper assignment so that no PC member is able to obtain illegitimate information 
about reports?
Parametric systems
like conference management systems 
can readily be formalized as 
first order (FO) transition systems where the attained states
of agents are given as a FO structure, 
i.e., a finite set of relations.
This approach was pioneered by
abstract state machines (ASMs)~\cite{gurevich2018evolving}, and has
found many practical applications, for example in the verification of
network protocols~\cite{padon2016ivy}, software defined networks~\cite{ball2014vericon}, and multi-agent
workflows~\cite{DBLP:conf/ccs/Finkbeiner0SZ17,ATVA16,DBLP:conf/csfw/0008SZ18}.  
FO transition systems rely on \emph{input} predicates to receive information
from the environment such as network events, interconnection topologies, 
or decisions of agents.
In addition to the externally provided inputs, there are also
\emph{internal} decisions that are made to ensure well-behaviour of the system.
This separation of input predicates into these two groups turns the underlying
transition system into a two-player \emph{game}.
In order to systematically explore possibilities of synthesizing 
message contents in protocols or strategies in workflows,
we generalize FO transition systems to FO games.

\begin{example}\label{e:leader_election}
	\Cref{fig:leaderelection} shows
	a slightly simplified version of the network leader election protocol from \cite{padon2016ivy}
    turned into a FO game.
    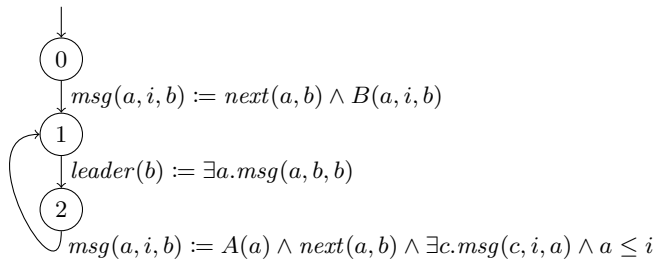
\begin{figure}[h!]
    \begin{tikzpicture}[
    node distance=1cm
]
    % \useasboundingbox (-1,0.6) rectangle (7.9,-4.8);

    \node[fotsnode, initial above, initial text={}] (0) at (0,0) {0};
    \node[fotsnode] (1) [below of=0] {1};
    \node[fotsnode] (2) [below of=1] {2};

    \path[->, draw]
        (0) edge node[right, align=left] {
            $\Msg(a,i,b) \coloneqq \Next(a,b) \land B(a,i,b)$
            } (1)
        (1) edge node[right, align=left] {
            $\Leader(b) \coloneqq \exists a. \Msg(a,b,b)$
        } (2);
        \path[->, draw] (2) edge [out=-90,in=180, looseness=2]
            node[pos=0.1, right, align=left] {
                $\Msg(a,i,b) \coloneqq A(a)\land \Next(a,b) \land
			\exists c. \Msg(c,i,a)\land a\leq i$
		% \exists b. \Next(a,b) \land \Msg(a,b) \land b \leq a$
            }  (1);
\end{tikzpicture}
    \vspace{-2em}
    \caption{\label{fig:leaderelection}
    FO safety game for the running leader election example}
    \end{figure}
    % It models a leader election protocol in a ring topology.
    The topology of the network, here a ring, is given by the predicates $\Next$ and $\leq$, 
    which are appropriately axiomatized.
    The participating agents communicate via messages through the predicate $\Msg$
    but are only allowed to send messages to the next agent in the ring topology.
    In the first step, agents can send any message (determined via the input predicate
    $B$) to their neighbor. Afterwards they check if they have received 
    a message containing their own id. If so, they declare themselves leader
    and add themselves to the $\Leader$ relation.
    Then, a subset of processes determined by the input predicate $A$ decides to send any id 
    to their next neighbor that they have received which is not exceeded by their own.

    At no point more than one process should have declared itself leader
    --- regardless of the size of the ring. 
    This property is enforced, e.g., if the initial message to be sent is given by the 
    id of the sending process itself, i.e., $B(a,i,b)$ is given by the literal $(i=a)$.
\qed
\end{example}

\begin{example}\label{e:easychair}
Consider the 
workflow of a conference management system as specified in
\cref{fig:easychair}.
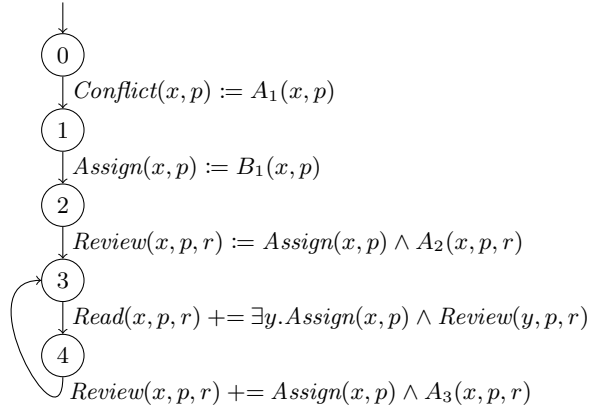
\begin{figure}[t]
    \begin{tikzpicture}[
	node distance=1cm
]
	\useasboundingbox (-1,0.6) rectangle (7.9,-4.9);

	\node[fotsnode, initial above, initial text={}] (0) at (0,0) {0};
	\node[fotsnode] (1) [below of=0] {1};
	\node[fotsnode] (2) [below of=1] {2};
	\node[fotsnode] (3) [below of=2] {3};
	\node[fotsnode] (4) [below of=3] {4};

	\path[->, draw]
		(0) edge node[right, align=left] {
			$\Conf(x,p) \coloneqq A_1(x,p)$
			} (1)
		(1) edge node[right, align=left] {
			$\Assign(x,p) \coloneqq B_1(x,p)$
		} (2)
		(2) edge node[right, align=left] {
			$\Review(x,p,r) \coloneqq \Assign(x,p) \land A_2(x,p,r)$
		} (3)
		(3) edge node[right, align=left] {
			$\Read(x,p,r) \wadd \exists y. \Assign(x,p) \land \Review(y,p,r)$
		} (4);

		\path[->, draw] (4) edge [out=-90,in=180, looseness=2]
			node[pos=0.1, right, align=left] {
				$\Review(x,p,r) \wadd \Assign(x,p) \land A_3(x,p,r)$
			}  (3);
\end{tikzpicture}
    \caption{\label{fig:easychair}
    FO safety game for the running conference management example}
\end{figure}
The specification maintains the binary predicates 
$\Conf$ and $\Assign$ together with the ternary predicates $\Review$ and $\Read$ to record
conflicts of interest between PC members and papers,
the paper assignment as well as the reports provided by PC members for papers.
After the initial declaration of conflicts of interest, PC members write reviews for
the papers they are assigned and update them after reading the other reviews to the
same paper. 
The predicates $A_1,A_2,A_3$ represent choices by PC members, 
while the predicate $B_1$ is under control of the PC chair.
The operator $\wadd$ adds tuples to a relation instead of 
replacing all contents. Specifically, $R\bar y \wadd \phi$ 
abbreviates $R \bar y \coloneqq R \bar y \lor \phi$.
\qed
\end{example}
One property to be checked in \cref{e:easychair} is that no PC member can learn 
anything about papers she has declared conflict with.
\emph{Noninterference} properties like this one can be formalized as
\emph{hyper-safety} properties, but
can be reduced to \emph{safety} properties of suitable \emph{self-compositions} of 
the system in question \cite{DBLP:conf/csfw/0008SZ18}.
This reduction is explained in \cref{a:ni}.
A plain safety property in this example would be, e.g.,
the more humble objective 
that no PC member $x$ is going to read a report on a paper $p$
which she herself has authored, i.e., 
\[
{\small\forall x,p,r. \neg(\Conf(x,p) \land \Read(x,p,r))}
\]
Obvious choices for $B_1$ to enforce this property are
\[
{\small
\begin{array}{lll}
B_1(x,p)	&{:=}	&\neg\Conf(x,p)\qquad\text{or}	\\
B_1(x,p)	&{:=}	&\tfalse
\end{array}
}
\]
The second choice is rather trivial. The first choice, on the other hand,
which happens to be the \emph{weakest} possible, represents a meaningful strategy.

In this paper, we therefore investigate cases where \emph{safety} is decidable and 
winning strategies for safety player are effectively computable and as weak as possible.
For FO transition systems as
specified by the Relational Modeling Language (RML)~\cite{padon2016ivy}, 
typed update commands are restricted to 
preserve Bernays-Sch\"onfinkel-Ramsey (also called \emph{effectively propositional}) formulas. 
As a consequence,
inductiveness of a universal invariant can be checked automatically.
We show that this observation can be extended to FO safety \emph{games} --- given 
that appropriate winning strategies for safety player are either provided 
or can be effectively constructed (see \cref{s:check}).
We also provide sufficient conditions under which a \emph{weakest}
such strategy can be constructed (see \cref{s:hilbert}).

The question arises whether a similar transfer of the
decidability of the logic to the decidability of the verification
problem is possible for other decidable fragments of FO logic. 
A both natural and useful candidate is \emph{monadic} logic.
Interestingly, this transfer is only possible for specific \emph{fragments} 
of monadic FO safety games, while in general safety is undecidable.
For FO safety games using arbitrary predicates, 
we restrict ourselves to FO universal invariants only,
since the safety properties, e.g., arising
from noninterference can be expressed in this fragment.
For universal invariants, we show how general methods for second order quantifier elimination 
can be instantiated to compute winning strategies.
Existential SO quantifier elimination, though, is not always possible.
Still, we provide a non-trivial class of universal invariants where 
optimal strategies can be synthesized.
In the general case and,
likewise, when existential FO quantifiers are introduced during game solving, 
we resort to \emph{abstraction} as in \cite{DBLP:conf/csfw/0008SZ18}.
This allows us to automatically construct strategies that guarantee safety
or, in the case of information-flow, to enforce noninterference.

The paper is organized as follows.
In \cref{s:fo-trans,s:games}, the notion of first-order safety games is introduced.
We prove that safety player indeed has a positional winning strategy, whenever 
the game is safe. 
We also prove that safety of \emph{finite} games is already inter-reducible to
SO predicate logic.
In \cref{s:monadic}, we consider the important
class of FO safety games where all predicates are either monadic 
or boolean flags. 
Despite the fact that this logic is decidable and admits SO
quantifier elimination, safety for this class is undecidable.
Nonetheless, we identify three subclasses of monadic games where
decidability is retained.
\Cref{s:check} proves that even
when a universally quantified FO candidate for
an inductive invariant of the safety game is already provided, 
checking whether or not the candidate invariant
is inductive, can be reduced to SO existential
quantifier elimination.
\Cref{s:hilbert} provides background techniques for
SO universal as well as existential quantifier elimination.
It proves that for universal FO formulas, 
the construction of a \emph{weakest} SO Hilbert choice operator can be reduced to 
SO quantifier elimination itself.
Moreover, it provides sufficient conditions when a universal invariant for 
a FO safety game can effectively proven inductive and a corresponding weakest strategy 
for safety player be extracted.
Based on the candidates for the second-order Hilbert choice operator from \cref{s:hilbert}, 
and abstraction techniques from \cite{DBLP:conf/csfw/0008SZ18}, a practical implementation
is presented in \cref{s:experimental} which allows to infer inductive invariants and
FO definable winning strategies for safety player.
Finally, \cref{s:refs} provides a more detailed comparison with related work while
\cref{s:conclusion} concludes.

\section{First-Order Transition Systems}\label{s:fo-trans}

Assume that we are given finite sets $\relsS$, $\relsI$, 
$\Const$ of relation symbols and constants,
respectively. 
A first-order (FO) transition system $\mathcal S$ (over $\relsS$, $\relsI$ and $\Const$)
consists of a control-flow graph $(V,E,\start)$ 
underlying $\mathcal S$
where $V$ is a finite set of program points, $\start\in V$ is the start point 
and $E$ is a finite set of edges between vertices in $V$. 
Each edge thereby is of the form $(v,\theta,v')$ where $\theta$ signifies
how the first-order structure for program point $v'$ is determined in terms of a
first-order structure at program point $v$.
Thus, $\theta$ is defined as a mapping which provides for each predicate $R\in\relsS$ of 
arity $r$, a first-order formula $R\theta$ with free variables from $\Const$ as well
as a dedicated sequence of fresh FO variables $\bar y=y_1\ldots y_r$. 
Each formula $R\theta$ may use FO quantification, equality or disequality literals as well as
predicates from $\relsS$.
Additionally, we allow occurrences of dedicated \emph{input} predicates from $\relsI$.
For convenience, we denote a substitution $\theta$ of predicates 
$R_1,\ldots,R_n$ with $\phi_1,\ldots,\phi_n$ by 
\[
{\small\theta = \{ R_1\bar y_1\coloneqq\phi_1,\ldots,R_n\bar y_n\coloneqq\phi_n\}}
\]
where $\bar y_j=y_1\ldots y_{r_j}$ are the formal parameters
of $R_i$ and may occur free in $\phi_i$.

\begin{example}
	In the example from \cref{fig:easychair}, the state predicates in $\relsS$ are 
	$\Conf$, $\Assign$, $\Review$ and $\Read$, while the input predicates $\relsI$ 
	consist of $A_1$, $A_2$, $A_3$ and $B_1$.
	As there are no global constants, $\Const$ is empty.
	For the edge from node $2$ to node $3$, $\theta$ maps $\Review$ to the 
	formula $\Assign(y_1,y_2)\land A_2(y_1,y_2,y_3)$ 
	and each other predicate $R$ from $\relsS$ to itself (applied to the appropriate list 
	of formal parameters).
	Thus, $\theta$ maps, e.g., $\Conf$ to $\Conf(y_1,y_2)$.
\qed
\end{example}

Let $U$ be some universe and $\rho:\Const\to U$ be a valuation of the globally free variables. 
Let $\relsSn$ denote the set of predicates with arity $n$.
A \emph{state} $s: \bigcup_{n\geq 0}\relsSn \times U^n \rightarrow \mathbb B$ 
is an evaluation of the predicates $\relsS$ by means of relations over $U$.
Let $\States_U$ denote the set of all states with universe $U$.
For an edge $(v,\theta,v')$, a valuation $\omega$ of the input predicates,
and states $s,s'$, there is a transition from $(v,s)$ to $(v',s')$ iff
for each predicate $R\in\relsS$ of arity $r$ together 
with a vector $\bar y= y_1\ldots y_r$ and an element $u\in U^r$
\[
s',\rho\oplus\{y\mapsto u\}\models Ry \text{ iff } s\oplus\omega,\rho\oplus\{y\mapsto u\}\models(R\theta)
\] holds. Here, the operator ``$\oplus$'' is meant to update the assignment in the left argument
with the variable/value pairs listed in the second argument.
The set of all pairs $((v,s),(v',s'))$ constructed in this way, 
constitute the \emph{transition relation}
$\Transitions_{U,\rho}$ of $\mathcal S$ (relative to universe $U$ and valuation $\rho$).
A finite \emph{trace} from $(v,s)$ to $(v',s')$
is a finite sequence $(v_0,s_0),\ldots,(v_n,s_n)$ with 
$(v,s)=(v_0,s_0)$ and $(v_n,s_n)=(v,s')$
such that for each $i=0,\ldots,n-1$,
$((v_i,s_i),(v_{i+1},s_{i+1}))\in\Transitions_{U,\rho}$ holds.
We denote the set of all finite traces of a transition system $\mathcal S$ as $\traces(\mathcal S)$.

\begin{example}
	Let us instantiate the running example from \cref{fig:easychair} for 
	the universe $\{ x_1, x_2, p_1, p_2, r_1 \}$. A possible 
	state attainable at node $2$ could have
	$\Conf = \{ (x_1, p_1) \}$, $\Assign = \{ (x_1,p_2), (x_2, p_1), (x_2, p_2) \}$ and all other relations empty.
	For the valuation $A_2 = \{ (x_2, p_2, r_1) \}$ of the input predicate, there would be a transition to node $3$ 
	and a state where $\Review = \{ (x_2, p_2, r_1) \}$, with $\Conf$ and $\Assign$ unchanged and $\Read$ still empty.
\qed
\end{example}

\section{First-Order Safety Games}\label{s:games}

For a first-order transition system, a FO \emph{assertion} is a mapping $I$ that
assigns to each program point $v\in V$ a FO formula $I[v]$ with relation symbols
from $\relsS$ and free variables from $\Const$.
Assume that additionally we are given a FO formula $\Init$ (also with relation symbols 
from $\relsS$ and free variables from $\Const$) describing the potential initial states. 
The assertion $I$ \emph{holds} if
for all universes $U$, all valuations 
$\rho$, all states $s$ with $s,\rho\models\Init$ and 
all finite traces $\tau$ from $(\start,s)$ to $(v,s')$,
we have that $s',\rho\models I[v]$.
In that case, we say that $I$ is an \emph{invariant} of the
transition system (w.r.t.\ the initial condition $\Init$).

\begin{example}\label{ex:release}
	For our running example from \cref{fig:easychair}, the initial condition 
	specifies that all relations $R$ in $\relsS$ are
	empty, i.e.,
	$
	\Init = \bigwedge_{R\in\relsS}\forall \bar y. \neg R\bar y
	$
	where we assume that the length of the sequence of variables $\bar y$
	matches the rank of the corresponding predicate $R$.
	Since the example assertion should hold everywhere, we have for every $u$,
	% $I[u]$ is 
	% \[
	$
	I[u] = \forall x,p,r. \neg(\Conf(x,p) \land \Read(x,p,r))
	$
	\qed
	% \eqno\qed
	% \]

\end{example}
We now generalize FO transition systems to \emph{FO safety games}, i.e., 2-player games
where reachability player $\pA$ aims at violating the given assertion $I$ while 
safety player $\pB$ tries to establish $I$ as an invariant.
To do so, player $\pA$ is able to choose the universe, 
which outgoing edges are chosen at a given node and
all interpretations of relations under his control.
Accordingly, we partition the set of input predicates $\relsI$ into subsets $\relsA$ and $\relsB$.
While player $\pB$ controls the valuation of the predicates in $\relsB$, player $\pA$ has 
control over the valuations of predicates in $\relsA$ 
as well as over the universe and the valuation of the FO variables in $\Const$.
For notational convenience, 
we assume that each substitution $\theta$ in
the control-flow graph contains at most one input predicate, and that all these are
distinct 
\footnote{
	In general, edges may use multiple input predicates of the same type.
	% An edge which uses multiple input predicates the same type,
	This can, however, always be simulated by a sequence of edges that stores the contents 
	of the input relations in auxiliary predicates from  $\relsS$ one by one, before realizing
	the substitution of the initial edge by means of the auxiliary predicates.}.
Also we consider a partition of the set $E$ of edges into the subsets $E_\pA$ and $E_\pB$ 
where the substitutions only at edges from $E_\pB$ may use predicates from $\relsB$. 
Edges in $E_\pA$ or $E_\pB$ will also be called $\pA$-edges or $\pB$-edges, respectively. 
For a particular universe $U$ and valuation $\rho$, 
a trace $\tau$ starting in some $(\start,s)$ with $s,\rho\models\Init$
and ending in some pair $(v,s')$ is considered a \emph{play}.
For a given play, player $\pA$ \emph{wins} iff $s',\rho\not\models I[v]$
and player $\pB$ wins otherwise.

A \emph{strategy} $\sigma$ for player $\pB$ is a mapping which for 
each $\pB$-edge $e=(u,\theta,v)$ with input predicate $B_e$ (of some arity $r$), 
each universe $U$, valuation $\rho$, each state $s$ and each play
$\tau$ reaching $(u,s)$, returns a relation $B'\subseteq U^r$.
Thus, $\sigma$ provides for each universe, the history of the play and the next edge 
controlled by $\pB$, a possible choice.
$\sigma$ is \emph{positional} or \emph{memoryless}, 
if it depends on the universe $U$, the valuation 
$\rho$, the state $s$ and the $\pB$-edge $(u,\theta,v)$ only.

A play $\tau$ \emph{conforms} to a strategy $\sigma$ for safety player $\pB$, if all input relations
at $\pB$-edges occurring in $\tau$ are chosen according to $\sigma$.
The strategy $\sigma$ is \emph{winning} for $\pB$ if $\pB$ wins all plays that conform to $\sigma$.
An FO safety game can be won by $\pB$ iff there exists a winning strategy for $\pB$.
In this case, the game is \emph{safe}.

\begin{example}
	In the running conference management \cref{e:easychair}, player $\pA$, who wants to reach a state where 
	the invariant from \cref{ex:release} is violated 
	(a state where someone reads a review to his own paper
	before the official release) has control over the predicates $A_1$, $A_2$, $A_3$ and
	thus provides the values for the predicates $\Conf$ and 	 
	$\Review$ and also determines how often the loop body is iterated.
	Player $\pB$ only has control over predicate $B_1$ which is used to determine the
	value of predicate $\Assign$.
	This particular game is \emph{safe}, and player $\pB$ has several winning strategies, 
	e.g., $B_1(x,p) \coloneqq \neg \Conf(x,p)$.
\qed
\end{example}

\begin{lemma}\label{l:positional}
If there exists a winning strategy for player $\pB$, then there also exists a winning strategy that is positional.
\end{lemma}
\begin{proof}
Once a universe $U$ is fixed, together with a
valuation $\rho$ of the globally free variables, the FO safety game $G$
turns into a reachability game $G_{U,\rho}$ where the positions are given by all pairs
$(v,s)\in V\times\States_U$
(controlled by reachability player $\pA$) together with all pairs
$(s,e)\in\States_U\times E$ controlled by safety player $\pB$ if $e\in E_\pB$ and by $\pA$ otherwise.
For an edge $e=(v,\theta,v')$ in $G$, $G_{U,\rho}$ contains all edges $(v,s)\to(s,e)$,
together with all edges $(s,e)\to(v',s')$ where $s'$
is a successor state of $s$ w.r.t.\ $e$ and $\rho$.

Let $\Init_{U,\rho}$ denote the set of all positions $(\start,s)$ where $s,\rho\models\Init$,
and $I_{U,\rho}$ the set of all positions $(v,s)$ where $s,\rho\models I[v]$
together with all positions $(s,e)$ where $s,\rho\models I[v]$ for edges $e$ starting in $v$.
Then $G_{U,\rho}$ is safe iff safety player $\pB$ has a strategy $\sigma_{U,\rho}$
to force each play started in some position $\Init_{U,\rho}$
to stay within the set $I_{U,\rho}$.
Assuming 
the axiom of choice for set theory,
the set of positions can be well-ordered. 
Therefore, the strategy $\sigma_{U,\rho}$ for safety player $\pB$ can be chosen positionally,
see, e.g., lemma 2.12 of \cite{Mazala1998}.
Putting all positional strategies $\sigma_{U,\rho}$ for safety player $\pB$
together, we obtain a single positional strategy for $\pB$.
\qed
\end{proof}

\noindent
In case the game is safe, 
we are interested in strategies that can be included into
the FO transition system itself, i.e., are themselves first-order definable.
\Cref{l:positional} as is, gives no clue
whether or not there is a winning strategy which is positional and can be expressed in 
FO logic, let alone be effectively computed.

\begin{theorem}\label{lem:expressiveness}
There exist safe FO safety games where no winning strategy is expressible in FO logic.  
\end{theorem}

\begin{proof}
Consider a game with $\relsS = \{E,R_1,R_2\}$, $\relsA = \{A_1,A_2\}$ and 
$\relsB = \{B_1\}$, performing three steps in sequence:
\[
{\small
\begin{array}{lll@{\quad}lll@{\quad}lll}
    E(x,y) &\coloneqq& A_1(x,y);	&
    R_1(x,y) &\coloneqq& B_1(x,y);	&
    R_2(x,y) &\coloneqq& A_2(x,y)
\end{array}
}\]
In this example, reachability player $\pA$ chooses an arbitrary relation $E$,
then safety player $\pB$ chooses $R_1$ and player $\pA$ chooses $R_2$. 
The assertion $I$ ensures that at the endpoint $R_1$ is at least 
the transitive closure of $E$ and $R_1$ is smaller or equal to $R_2$
(provided $\pA$ chose $R_2$ to include the closure of $E$)
i.e.,
\[
{\small
		\textsf{closure}(R_2, E) \implies (\textsf{closure}(R_1, E) \land
				\forall x,y.  (R_1(x,y) \implies R_2(x,y)))}
\]
where ${\small\textsf{closure}(R, E)}$ is given by
$
	{\small\forall x,y. R(x,y) \leftarrow (E(x,y)\ \lor \exists z. R(x,z) \land E(z,y))}
$.
The only winning strategy for safety player $\pB$ (choosing $R_1$) is to select the smallest relation 
$R_1$ satisfying $\mathit{closure}(R_1, E)$, which is the transitive closure of $E$.
In this case, no matter what reachability player $\pA$ chooses for $R_2$, 
safety player $\pB$ wins, but the winning strategy for $\pB$ is not expressible in FO logic.
\qed
\end{proof}

\noindent
Despite this negative result, effective
means are sought for of computing FO definable strategies, whenever they exist.
In order to do so, we rely on a \emph{weakest precondition} operator $\sem{e}^\top$
corresponding to edge $e= (u,\theta,v)$ of the control-flow graph of a FO safety game $\T$ 
by
\[
{\small\sem{e}^\top \Psi =\left\{
        \begin{array}{lll}
        \forall A_e.(\Psi\theta) &\text{if}\;e\;\pA\text{-edge}  \\
        \exists B_e.(\Psi\theta) &\text{if}\;e\;\pB\text{-edge}  \\
        \end{array}\right.}
\]
The weakest pre-condition operator captures the minimal requirement at the start point of an edge 
to meet the post-condition $\Psi$ at the end point.
That operator allows to define the following iteration:
Let $\T$ denote a game and $I$ an assertion. 
For $h\geq 0$, let the assignment $\Psi^{(h)}$
of program points $v$ to formulas be
\begin{eqnarray}
{\small
\begin{array}{lll}
\Psi^{(0)}[v]	&=&	I[v]	\\
\Psi^{(h)}[v]	&=&	\Psi^{(h-1)}[v]\;	
% \\
% &&	\begin{array}[t]{@{}l}
	\wedge\;\bigwedge_{e\in \mathit{out(v)}}\sem{e}^\top \Psi^{(h-1)}
	\quad	\text{for}\; h>0
	% \end{array}
\end{array} \label{def:fixpoint}
}
\end{eqnarray}
where $\mathit{out}(v)$ are the outgoing edges of node $v$.
Then the following holds:

\begin{restatable}{theorem}{tapprox}
 \label{tapprox}
% \begin{lemma}\label{l:approx}
A FO safety game $\T$ is safe iff
% There is a winning strategy for the safety player $\pB$ iff
% {\small
% \begin{eqnarray}
$
\Init\implies\Psi^{(h)}[\start]	\label{eq:init}
$
% \end{eqnarray}
% }
holds for all $h\geq 0$. 
\end{restatable}

\noindent
The proof can be found in \cref{a:games}.
The characterization of safety due to \cref{tapprox} is precise --- but may require
to construct infinitely many $\Psi^{(h)}$.
Whenever, though, the safety game $\T$ is \emph{finite}, i.e., 
the underlying control-flow graph of $G$ is acyclic,
then $G$ is safe iff $\Init\implies\Psi^{(h)}[\start]$ where
$h$ equals the length of the longest
path in the control-flow graph of $G$ starting in $\start$.
As a result, we get that \emph{finite} first order safety games are as powerful as second order logic.
\begin{theorem}\label{t:finite}
Deciding a finite FO safety game with predicates from $\relsS$ is inter-reducible to
satisfiability of SO formulas with predicates from $\relsS$.
\end{theorem}

\begin{proof}
We already showed that solving a finite FO safety game can be achieved by
solving the SO formula $\psi^{(h)}$ for some sufficiently large $h$.
For the reverse implication,
consider an arbitrary closed formula $\phi$ in SO Logic.
W.l.o.g., assume that $\phi$ has no function symbols and is in prenex normal form
where no SO Quantifier falls into the scope of a FO quantifier 
\cite{leivant1994higher}. 
Thus, $\phi$ is of the form $Q_1 C_1 \ldots Q_n C_n.\ \psi$ where all $Q_n$ are SO quantifiers and $\psi$ is a relational formula in FO logic.

We then construct a FO safety game $\T$ as follows. 
The set $\relsS$ of predicates consists of all predicates that occur freely in $\phi$ 
together with copies $R'_{i}$ of all quantified relations $C_i$.
The control-flow graph consists of $n+1$ nodes $v_0, \ldots, v_{n}$,
together with edges $(v_{i-1},\theta_i,v_{i})$ for $i=1,\ldots,n$.
Thus, the maximal length of any path is exactly $n$.
An edge $e = (v_{i-1},\theta_i,v_i)$ 
is used to simulate the quantifier $Q_{i}C_i$. 
The substitution $\theta_i$ is the identity on all predicates from $\relsS$ 
except $R'_i$ which is mapped to $C_i$.
If $Q_{i}$ is a universal quantifier, 
$C_{i}$ is included into $\relsA$, and $e$ is an $\pA$-edge.
Similarly, if $Q_{i}$ is existential, $C_{i}$ is included into $\relsB$ and 
$e$ is a $\pB$-edge.
Assume that $\psi'$ is obtained from $\psi$ by replacing every relation $R_i$ with $R'_i$.
As FO assertion $I$, we then use 
$I[v_i] = \ttrue$ for $i=0,\ldots,n-1$ and $I[v_n]=\psi'$.
Then $\Psi^{(n)}[v_n] = \phi$. 
Accordingly for $\Init = \ttrue$, player $\pB$ can win the game iff $\phi$ is universally true.
\qed
\end{proof}

\noindent
Theorem \ref{t:finite} implies that a FO definable winning strategy for safety player $\pB$ 
(if it exists)
can be constructed whenever the SO quantifiers introduced by the choices of the respective
players can be eliminated. 
Theorem \ref{t:finite}, though, 
gives no clue \emph{how} to decide whether or not safety player $\pB$ has
a winning strategy and if so, whether it can be effectively represented.

\section{Monadic FO Safety Games}\label{s:monadic}

\noindent
Assuming that the universe is finite and bounded in size by some $h \geq 0$, 
then FO games reduce to finite games (of tremendous size, though). 
This means that, at least in principle, both checking of invariants as well as the construction of
a winning strategy (in case that the game is safe) is effectively possible.
A more complicated scenario arises 
when the universe consists of several disjoints \emph{sorts}
of which some are bounded in size and some are unbounded.

We will now consider the special case where each predicate 
has at most one argument which takes elements of an 
unbounded sort.
In the conference management example, we could, e.g., assume that 
PC members, papers and reports constitute disjoint sorts of bounded cardinalities, 
while the number of (versions of) reviews is unbounded.
By encoding the tuples of elements of finite sorts into predicate names,
we obtain FO games where all predicates are either nullary or \emph{monadic}.
Monadic FO logic is remarkable since satisfiability of formulas in that logic is decidable,
and monadic SO quantifiers can be effectively eliminated \cite{behmann1922beitrage,wernhard2015second}.
Due to \cref{t:finite}, we therefore conclude for \emph{finite} monadic safety games
that safety is decidable. Moreover, in case the game is safe, 
a positional winning strategy for safety player $\pB$ can be effectively computed.

Monadic safety games which are not finite, 
turn out to be 
very close in expressive power to \emph{multi-counter machines},
for which reachability is undecidable \cite{holzer2011complexity,rosenberg1966multi}. 
The first statement of the following theorem has been communicated to us by
Igor Walukiewicz:

\begin{restatable}{theorem}{thmcounterex}
\label{t:counter}
For monadic safety games, safety
is undecidable when one of the following conditions is met:
\begin{enumerate}
\item   there are both $\pA$-edges as well as $\pB$-edges;
\item	there are $\pA$-edges and substitutions with equalities or disequalities;
\item	there are $\pB$-edges and substitutions with equalities or disequalities.
\end{enumerate}
\end{restatable}

\noindent
The proof of statement (1) is by using monadic predicates to simulate the counters of
a multi-counter machine. Statements (2) and (3) follow from the observation 
that one player in this simulation can be replaced by substitutions using equality or disequality
literals (see \cref{a:counter} for details of the simulation). 

There are, though, interesting cases that do not fall into the listed classes and can
be effectively decided.
Let us first consider monadic safety games where no predicate is under the control of either player, 
i.e., $\relsA = \relsB = \emptyset$, but both equalities and disequalities are allowed.
Then, safety of the game collapses to the question if player $\pA$ can pick universe and 
control-flow such that the assertion is violated at some point.
For this case, we show that the conjunction of preconditions from \cref{def:fixpoint} necessarily stabilizes.

\begin{restatable}{theorem}{thmmonadicplain}
\label{t:monadic_plain}
Assume that $\T$ is a monadic safety game,
possibly containing equalities and/or disequalities
with $\relsA = \relsB = \emptyset$.
Then for some $h\geq 0$, $\Psi^{(h)}=\Psi^{(h+1)}$.
Therefore, safety of $\T$ is decidable.
\end{restatable}

\noindent
Theorem \ref{t:monadic_plain} relies on the observation
that when applying substitutions alone, i.e., without additional SO quantification,
the number of equalities and disequalities involving FO variables,
remains bounded. 
Our proof relies on variants of the \emph{counting quantifier normal form} 
for monadic FO formulas \cite{behmann1922beitrage}
(see \cref{a:monadic_plain}).

\label{ss:monadic_A}
Interestingly, decidability is also retained for assertions $I$ that only contain disequalities, 
if no equalities between bound variables are
introduced during the weakest precondition computation.
This can only be guaranteed if safety player $\pB$ does not have control over any predicates.
\footnote{The simulation in \cref{a:counter} shows how predicates under the control of player $\pB$ can be used to introduce equalities through SO existential quantifier elimination.}

\begin{restatable}{theorem}{thmmonoa}
\label{t:mono_A}\label{t:mono-mc}
Assume that $\T$ is a monadic safety game without $\pB$-edges (i.e. $\relsB = \emptyset$) and 
\begin{enumerate}
	\item there are no disequalities between bound variables in $I$,
	\item in all literals $x = y$ or $x \neq y$ in $\Init$ and substitutions $\theta$, $x\in\Const$ or $y\in\Const$.
\end{enumerate}
Then it is decidable whether $\T$ is safe.
\end{restatable}

The proof is based on the following observation:
Assume that $\Const$ is a set of variables of cardinality $d$, and
formulas $\phi_1,\phi_2$ have free variables only from $\Const$.
If $\phi_1,\phi_2$ contain no disequalities between bound variables,
then $\phi_1,\phi_2$ are equivalent for all models and all valuations $\rho$
iff they are equivalent for models and valuations with \emph{multiplicity} exceeding $d$.
Here, the \emph{multiplicity} $\mu(s)$ of a model $s$
is the minimal cardinality of a 
non-empty equivalence class of $U$ w.r.t.\ \emph{indistinguishability}.
We call two elements $u,u'$ of the universe $U$ indistinguishable
in a model $s$ iff
$(s,\{x\mapsto u\}\models Rx)\leftrightarrow (s,\{x\mapsto u'\}\models Rx)$
for all relations $R$. 
Then, when computing $\Psi^{(h)}$, we use an \emph{abstraction}
by formulas without equalities, which is shown to be a \emph{weakest strengthening}
(see \cref{proof_mono_A}). 

Analogously, decidability is retained for assertions that only contain positive equalities
if there are no disequalities introduced during the weakest precondition computation.
This is only the case when $\relsA =\emptyset$, i.e.,
reachability player $\pA$ only selects
universe and control-flow path. 
As a consequence, we obtain:

\begin{restatable}{theorem}{thmmonob}
\label{t:mono_B}
Assume that $\T$ is a monadic safety game without $\pA$-edges where
\begin{enumerate}
	\item there are no equalities between bound variables in $I$,
	\item in all literals $x = y,x\neq y$ in $\Init$ and substitutions $\theta$,
	either $x\in\Const$ or $y\in\Const$.
\end{enumerate}
Then it is decidable whether $\T$ is safe.
\end{restatable}

The proof is analogous to the proof of theorem \ref{t:mono-mc}
where the abstraction of equalities now is replaced with an abstraction
of disequalities (see \cref{proof_mono_B}). 
In summary, we have shown that even though monadic logic is decidable,
2-player monadic FO safety games are undecidable in general.
However, for games where one of the players does not
choose interpretations for any relation,
decidability can be salvaged if the safety condition has acceptable 
equality/disequality literals only and neither $\Init$, nor the transition relation 
introduce further equality/disequality literals between bound variables.

\section{Proving Invariants Inductive}\label{s:check}

Even though the general problem of verification is already hard for monadic 
FO games, there are useful incomplete algorithms to
still prove general FO safety games safe.
One approach for verifying infinite state systems is to come up with
a candidate invariant which then is proven \emph{inductive} (see, e.g., \cite{padon2016ivy}).
This idea can be extended to \emph{safety games} where, additionally
strategies must either be provided or extracted.

In the context of FO safety games, 
an invariant $\Psi$ is called \emph{inductive} iff for all edges $e=(u,\theta,v)$,
$\Psi[u]\implies\sem{e}^\top(\Psi[v])$ holds.
\ignore{
\[
{\small
\begin{array}{lll}
\Psi[v]&\implies&
	\forall A_e.(\Psi[v']\theta)\quad\text{if}\;e\in E_\pA	\\
\Psi[v]&\implies&
	\exists B_e.(\Psi[v']\theta)\quad\text{if}\;e\in E_\pB	\\
\end{array}}
\]
}
We have:

\begin{lemma}\label{l:inductive}
Assume that $\Psi$ is inductive, and $\Psi[v]\implies I[v]$ for all nodes $v$. 
Then 
\begin{enumerate}
\item	For all $h\geq 0$, $\Psi[v]\implies\Psi^{(h)}[v]$;
\item	The game $G$ is safe, whenever $\Init\implies\Psi[\start]$ holds.
\end{enumerate}
\end{lemma}
We remark that, under the assumptions of \cref{l:inductive}, a \emph{positional} 
winning strategy $\sigma$ for safety player $\pB$ \emph{exists}.
Checking an FO safety game $\T$ for safety thus boils down to the following 
tasks:
\begin{enumerate}
\item	Come up with a candidate invariant $\Psi$ so that 
	\begin{itemize}
	\item	$\Psi[v]\implies I[v]$ for all nodes $v$, and
	\item	$\Init\implies\Psi[\start]$ hold;
	\end{itemize}
\item	Come up with a strategy $\sigma$ which assigns
	some FO formula to each predicate in $\relsB$;
\item	Prove that $\Psi$ is inductive for the FO transition system $\T\sigma$
	which is obtained from $\T$ by substituting each occurrence of $B$ 
	with $\sigma(B)$ for all $B\in\relsB$.
\end{enumerate}
%%%
For monadic FO safety games, we thereby obtain:

\begin{theorem}\label{t:monadic_check}
Assume that $\T$ is a monadic FO safety game with initial condition $\Init$ and assertion $I$.
Assume further that $\Psi$ is a monadic FO invariant, i.e., maps each program point to a
monadic formula.
Then the following holds:
\begin{enumerate}
\item	It is decidable whether $\Init\implies\Psi[\start]$ as well as
	$\Psi[v]\implies I[v]$ holds for each program point $v$;
\item	It is decidable whether $\Psi$ is inductive, and if so, an FO definable
	strategy $\sigma$ can be constructed which upholds $\Psi$.
\end{enumerate}
\end{theorem}
The proof is by showing that all formulas fall into a decidable fragment ---
in this case Monadic Second Order logic.
\noindent
A \emph{monadic} FO safety game can thus be proven safe by providing an appropriate 
monadic FO invariant $\Psi$: the winning strategy itself can be effectively computed.

Another important instance is when the candidate invariant $\Psi$ as well as $I$ 
consists of universal FO formulas only,
while $\Init$ is in the \emph{Bernays-Schönfinkel-Ramsey} (BSR) fragment\footnote{
	The Bernays-Schönfinkel-Ramsey fragment contains all formulas of First Order Logic that
	have a quantifier prefix of $\exists^*\forall^*$ and do not contain function symbols.
	Satisfiability of formulas in BSR is known to be decidable \cite{ramsey2009problem}.
}.

\begin{theorem}\label{t:check}
Let $\T$ denote a safety game
where each substitution $\theta$ occurring at edges of the control-flow graph
uses non-nested FO quantifiers only.
Let $\Psi$ denote a universal FO invariant for $\T$, i.e.,
$\Psi[v]$ is a universal FO formula for each node $v$.
\begin{enumerate}
\item	It is decidable whether $\Init\implies\Psi[\start]$ as well as
	$\Psi[v]\implies I[v]$ holds for each program point $v$;
\item	Assume that no $B\in\relsB$ occurs in the scope of an existential FO quantifier,
	and $\sigma$ is a strategy which provides 
	a universal FO formula for each $B\in\relsB$.
Then it is decidable whether or not $\Psi$ is inductive for $\T\sigma$.
\end{enumerate}
\end{theorem}

\noindent
The proof is by showing that all mentioned formulas can be solved by
checking satisfiability of a formula in the decidable fragment $\eufol$.
\Cref{t:check} states that (under mild restrictions on the substitutions
occurring at $\pB$-edges), the candidate invariant $\Psi$ can be checked for 
inductiveness --- at least when
a positional strategy of $\pB$ is provided which is expressed by means of universal FO formulas.
In particular, this implies decidability for the case when the set $E_\pB$ is empty.
The proof works by showing that all verification conditions fall into the BSR fragment of FO Logic.
For the verification of inductive invariants for FO transition systems
(no $\pB$ edges),
the \textsc{Ivy} system essentially relies on the observations summarized in 
\cref{t:check} \cite{padon2016ivy}.

Besides finding promising strategies $\sigma$,
the question remains how for a given assertion $I$ a suitable \emph{inductive} invariant can be
inferred. One option is to iteratively compute the sequence $\Psi^{(h)}, h\geq 0$ as in
\eqref{def:fixpoint}.
In general that iteration may never reach a fixpoint.
Here, however, FO definability implies termination:

\begin{theorem}\label{t:goedel}
Assume that for all program points $u$ and $h\geq 0$, $\Psi^{(h)}[u]$ is FO definable
as well as the infinite conjunction $\bigwedge_{h\geq 0}\Psi^{(h)}[u]$.
Then there exists some $m\geq 0$ such that
$\Psi^{(m)}=\Psi^{(m+k)}$ holds for each $k\geq 0$.
Thus, $\bigwedge_{h\geq 0}\Psi^{(h)}[u] = \Psi^{(m)}[u]$ for all $u$.
\end{theorem}

\begin{proof}
Let $\phi_u$ denote the first order formula which is equivalent $\bigwedge_{h\geq 0}\Psi^{(h)}[u]$.
In particular, this means that $\phi_u\implies\Psi^{(h)}[u]$ for each $h\geq 0$.
On the other hand, we know that $\bigwedge_{h\geq 0}\Psi^{(h)}[u]$ implies $\phi_u$.
Since $\phi_u$ as well as each $\Psi^{(h)}[u]$ are assumed to be FO definable,
it follows from Gödel's compactness theorem that there is a \emph{finite}
subset $J\subseteq\mathcal{N}$ such that $\bigwedge_{h\in J}\Psi^{(h)}[u]$ implies $\phi_u$.
Let $m$ be the maximal element in $J$.
Then, $\bigwedge_{h\in J}\Psi^{(h)}[u] = \Psi^{(m)}[u]$  since the $\Psi^{(h)}[u]$
form a decreasing sequence of formulas.
Together, this proves that $\phi_u$ is equivalent to $\Psi^{(m)}[u]$.
\qed
\end{proof}
\Cref{t:goedel} proves that if there exists an inductive invariant proving a given FO
game safe, then fixpoint iteration will definitely terminate and find it.  

For the case of \emph{monadic} FO safety games, 
this means that the corresponding infinite conjunction 
is not always FO definable --- otherwise decidability would follow.
In general, not every invariant $I$ can be strengthened to an
inductive $\Psi$, and universal strategies need not be sufficient to win a 
universal safety game.
Nonetheless, there is a variety of non-trivial cases where existential SO quantifiers can
be effectively eliminated,  
e.g., by Second Order quantifier elimination algorithms 
SCAN or \DLS\ (see the overview in \cite{gabbay2008second}).
In our case, in addition to plain elimination we need an explicit construction of the 
corresponding strategy, expressed as a FO formula.
We remark that following \cref{t:check},
it is not necessary to perform \emph{exact} quantifier elimination: 
instead, a sufficiently weak \emph{strengthening} may suffice.
Techniques for such \emph{approximate} SO existential quantifier elimination are provided
in the next section.

\section{%SO Quantifier Elimination and 
	Hilbert's Choice Operator for Second Order Quantifiers}\label{s:hilbert}
\label{s:qe}

In this section, we concentrate on formulas with universal FO quantifiers only.
First, we recall 
% that for these, SO universal quantifiers can always be eliminated.
% This is based on 
the following observation:

\begin{lemma}[see Fact 1, \cite{DBLP:conf/csfw/0008SZ18}]\label{l:forall}
Consider a disjunction $c$ of the form
\[
{\small F\vee \begin{array}{l}
		\bigvee_{i=1}^kA\bar z_i\vee \bigvee_{j=1}^l\neg A\bar z'_i
		\end{array}}
\]
for some formula $F$ without occurrences of predicate $A$. Then 
% \[{\small
$
\forall A.c$ is equivalent to $\begin{array}{l}
	F\vee\bigvee_{i,j} (\bar z_i=\bar z'_j)
	\end{array}
$
% }\]
for sequences of variables $\bar z_i=z_{i1}\ldots z_{ir}$, $\bar z'_j=z'_{j1}\ldots z'_{jr}$, 
where $\bar z_i=\bar z_j$ is an abbreviation for $\bigwedge_{k=1}^r z_{ik}=z'_{jk}$.
\qed
\end{lemma}
\noindent
As a consequence, universal SO quantification can always be eliminated from 
universal formulas.

\begin{example}\label{ex:SOQE}
% \itodo{this is not the running example assertion, but this one uses SOQE. do we want to use the same one everywhere?}
Consider the assertion $I=  \forall x,p,r.\neg(\Conf(x,p)\land\Review(x,p,r))$
% from the introductory example
and substitution $\theta$ from the edge between program points 2 and 3 in \cref{fig:easychair},
given by $\Review(x,p,r) \coloneqq \Assign(x,p)\wedge A_3(x,p,r)$
and $\Conf(x,p) \coloneqq \Conf(x,p)$.
Since $I\theta$ contains only negative occurrences of $A_3$,
%SO universal quantifier elimination is particularly simple:
we obtain:
\[
{\small\begin{aligned}
\forall A_3.(I\theta)\;
=\;& \forall x,p,r.\forall A_3.\neg\Conf(x,p)\vee\neg\Assign(x,p)\vee \neg A_3(x,p,r)	\\
=\;& \forall x,p,r. \neg\Conf(x,p)\vee\neg\Assign(x,p)
	\label{def:I2}
\end{aligned}}
\eqno\qed
\]
\end{example}
As we have seen in section \ref{s:check}, 
checking whether a universal FO invariant is inductive
can be reduced to SO existential quantifier elimination.
%
% In the general case where predicates are neither nullary nor monadic,
% SO quantifier elimination may not always be possible --- even if we restrict ourselves to
% formulas with universal FO quantification only.
%
While universal SO quantifiers can always be eliminated in formulas with universal 
FO quantifiers only,
% (see, e.g., \cite{...}), 
this is not necessarily the case for existential SO quantifiers.
As already observed by Ackermann \cite{Ackermann1935}, the formula
\[
{\small\exists B.\,Ba\wedge\neg Bb\wedge \forall x,y.\neg Bx\vee\neg Rxy\vee By}
\]
expresses that $b$ is not reachable from $a$ via the edge relation $R$
and thus cannot be expressed in FO logic. 
%
% Thus, a general construction for eliminating SO existential quantification
% from otherwise FO universal formulas cannot exist. 
%
This negative result, though, does not exclude that in a variety of meaningful cases,
equivalent FO formulas can be constructed.
For formulas with universal FO quantifiers only, 
we provide a simplified algorithm for existential SO quantifier elimination.
Moreover, we show that the construction of a \emph{weakest} SO \emph{Hilbert choice operator}
can be reduced to existential SO quantifier elimination itself.
In terms of FO safety games, the latter operator enables us to 
extract \emph{weakest} winning strategies for safety player $\pB$.
%
% we introduce a sequence of candidate substitutions to approximate this operator.
% In particular, we show that 
% for these formulas,
% the construction of a \emph{weakest} SO Hilbert choice operator 
% We also indicate large classes of formulas where SO quantifier elimination is possible. 
%
For an in-depth treatment on SO existential quantifier elimination,
% and a detailed exposition of the (partial) SCAN as well as the \DLS\ algorithm, 
we refer to \cite{gabbay2008second}.
%
%%%

% As already observed by Ackermann \cite{Ackermann1935}, 
% \emph{existential} SO quantification 
% for an universal FO formula 
% need not always be equivalent to some FO formula.
% \itodo{insert ackermann example}
%
%
% Nonetheless, SO existential quantifier elimination from FO universal formulas
% is possible in several non-trivial cases.
%
% Refer to textbook and the scan algorithm!!

Let $\phi$ denote some universally quantified formula, possibly containing
a predicate $B$ of arity $r$. Let $\bar y = y_1\ldots y_r$, and $\bar y'=y'_1\ldots y'_r$.
We remark that for \emph{any} formula $\psi$ with free variables in $y$,
$\phi[\psi/B] \implies \exists B.\phi$ holds.
Here, this SO substitution means that every literal $B\bar z$ and every literal $\neg B\bar z'$ is
replaced with $\psi[\bar z/\bar y]$ and $\neg\psi[\bar z'/\bar y]$, respectively.
Let $H_{B,\phi}$ denote the set of \emph{all} FO formulas $\psi$ such that 
$\exists B.\phi$ is equivalent to $\phi[\psi/B]$. 
A general construction for $B$ and $\phi$ 
(at least from some suitably restricted class of formulas) 
of some FO formula $\psi\in H_{B,\phi}$
% such that $\phi[\psi/B] = \exists B.\phi$
is an instance of Hilbert's (second-order) \emph{choice} operator.
If it exists, we write $\psi=\mathcal{H_B}(\phi)$.
In order to better understand the construction of such operators,
we prefer to consider universal FO formulas in \emph{normal form}.

\begin{restatable}{lemma}{lemnormalform}
\label{l:normal}
Every universal FO formula $\phi$ possibly containing occurrences of $B$
% $\forall x.\phi'$, $\phi'$ quantifierfree, $x$ a sequence of
% variables distinct from $y,y'$, 
is equivalent to a formula 
{\small
\begin{equation}
{\small E\wedge(\forall \bar y.F\vee B\bar y)\wedge(\forall \bar y'.G\vee\neg B\bar y')\wedge(\forall \bar y\bar y'.H\vee B\bar y\vee\neg B\bar y')}
\label{def:normal}
\end{equation}}
where $E,F,G,H$ are universal formulas without 
% occurrences of 
$B$.
% $E$ has free variables from $x$, 
% $G$ has free variables from $y$, and
% $H$ has free variables from $x,y,y'$.
\end{restatable}

\noindent
The corresponding construction is provided in \cref{l:normalproof}. 
For that, disequalities between variables, and fresh
auxiliary variables $\bar y$ and $\bar y'$ are introduced,
where the sequence $\bar y'$ is only required when
both positive and negative $B$ literals occur within the same clause.
%
% The given normal form for formulas $\forall x.\phi'$ is closely related to the
% \emph{Eliminationshauptform} proposed by Behmann \cite{...} for monadic formulas ---
% with the notable difference that we allow clauses containing occurrences both of positive
% and negative $B$-literals.
% 
% CITATION OF GABBAY AS WELL?
% 
In case these are missing, the formula is said to be in \emph{simple} normal form.
%
% In fact, it is hopefully proven somewhere that it does not even exist in general!?
%
For that case, Ackermann's lemma applies:

\begin{lemma}[Ackermann's lemma \cite{Ackermann1935}]\label{l:choice}
% Ackermann's lemma
Assume that $\phi$ is in simple normal form
% {\small
% \begin{equation*}
$
% \phi = 
E\wedge(\forall \bar y.F\vee B\bar y)\wedge(\forall \bar y.G\vee\neg B\bar y)
$.
% \label{def:simple_normal}
% \end{equation*}}
Then we have:
\begin{enumerate}
\item	$\exists B.\phi = E\wedge(\forall \bar y.F\vee G)$;
\item	For every FO formula $\psi$,
	$\exists B.\phi = \phi[\psi/B]$ iff
	$(E\land\neg F)\implies\psi$ and
	$\psi\implies(\neg E\lor G)$.
	\qed
\end{enumerate}
\end{lemma}

\noindent
% to \cref{l:choice}, 
For formulas in simple normal form
a Hilbert choice operator thus is given by:
\begin{equation}
{\small {\mathcal{H}}_B\phi	= \neg E\lor G}	\label{def:simple_hilbert}
\end{equation}
% In general, though, the formula $\exists B.\phi$ need not have an equivalent 
% first-order formula.
% The given definition $\mathcal{H}_B\phi$ provides us with 
--- which is the \emph{weakest} $\psi$ 
for which $\exists B.\phi$ is equivalent to $\phi[\psi/B]$.

\begin{example}\label{ex:weakest}
% \itodo{update for new example}
For the invariant from \cref{ex:release},
% as arising from the running example before the last statement. 
the weakest precondition w.r.t.\ the second statement amounts to:
% \[
$
{\small\exists B_1.\forall x,p.\neg\Conf(x,p)\vee\neg B_1(x,p)}
$
which is $\ttrue$ for any formula $\Psi$ for $B_1$ (with free $x,p$) implying
$\neg\Conf(x,p)$.
\qed
\end{example}

\noindent The \emph{strongest} solution according to \cref{ex:weakest} thus is that the PC chair 
decides to
assign papers to \emph{no} PC member. 
While guaranteeing safety, this choice is not very useful. 
The \emph{weakest} choice on the other hand, provides us here with a decent
strategy. In the following we therefore will aim at constructing as \emph{weak} strategies
as possible.
% situations where weakest choices can be constructed.

Ackermann's Lemma gives rise to a nontrivial class of safety games where existential 
SO quantifier elimination succeeds.
% to our setting, we find that a Hilberts Choice Operator exists
% for several fragments of First Order Safety Games:
% For a given substitution $\theta$,
% FO safety game $G$ and a predicate $B \in \relsB$,
% 
We call $B \in \relsB$ \emph{ackermannian} in the substitution $\theta$ iff 
for every predicate $R \in \relsS$, if $\theta(R)$ contains $B$ literals, 
$\theta(R)$ is quantifierfree and its CNF does not contain clauses 
with both positive and negative $B$ literals.

\begin{theorem}\label{t:positive}
Assume we are given a FO Safety Game $\T$ where all substitutions contain nonnested quantifiers only, and a universal inductive invariant $\Psi$.
Assume further that the following holds:
\begin{enumerate}
\item 	All predicates $B$ under the control of safety player $\pB$ are ackermannian in
		all substitutions $\theta$;
\item	For every $\pB$-edge $e=(u,\theta,v)$, 
	every clause of $\Psi[v]$ contains at most one literal with a
	predicate $R$ where $\theta(R)$ has a predicate from $\relsB$.
\end{enumerate}
Then the weakest FO strategy for safety player $\pB$ can be effectively computed
for which $\Psi$ is inductive.
\end{theorem}

\begin{proof}
Consider an edge $(u,\theta,v)$ where the predicate $B$ under control of safety player occurs
in $\theta$. Assume that $\Psi[v]=\forall\bar z.\Psi'$ where $\Psi'$ is quantifierfree and
in conjunctive normal form. Since $\theta$ is ackermannian and due to the restrictions given
for $\Psi'$, $\Psi'$ can be written as $\Psi' = \Psi_0\land\Psi_1\wedge\Psi_2$ where
$\Psi_0,\Psi_1,\Psi_2$ are the conjunctions of clauses $c$ of $\Psi'$  where $\theta(c)$
contains none, only positive or only negative occurrences of $B$-literals, respectively.
In particular, $\theta(\Psi_0)$ is a FO formula without nested quantifiers.
The formula $\theta(\Psi_1)$ is equivalent to a conjunction of formulas of the form
% \[
$
F\vee B\bar y_1\vee \ldots B\bar y_r
$
% \]
where $F$ has non-nested quantifiers only, which thus are equivalent to 
\[
\forall\bar y.F\vee(\bar y_1\neq\bar y)\wedge\ldots\wedge(\bar y_r\neq\bar y)\vee B\bar y
\]
Likewise, 
$\theta(\Psi_2)$ is equivalent to a conjunction of formulas of the form
% \[
$
G\vee\neg B\bar y_1\vee \ldots\neg B\bar y_r
$
% \]
where $G$ has non-nested quantifiers only, which thus are equivalent to 
\[
\forall\bar y.G\vee(\bar y_1\neq\bar y)\wedge\ldots\wedge(\bar y_r\neq\bar y)\vee\neg B\bar y
\]
Therefore, we can apply Ackermann's lemma to obtain a formula $\bar\Psi$ in \uefol
equivalent to $\exists B.\,\theta(\Psi[v])$. Likewise, we obtain a weakest FO formula 
$\phi$ for $B$ so that $\theta(\Psi[v])[\phi/B] = \bar\Psi$.
Since $\Psi[u]$ only contains universal quantifiers, $\Psi[u]\implies\bar\Psi$ 
is effectively decidable.
\qed
\end{proof}
\begin{example}
Consider the leader election protocol from \cref{e:leader_election},
together with the invariant from \cite{padon2016ivy}. Therein, the predicate
$B$ is ackermannian, and
$\Msg$ appears once in two different clauses of the invariant.
Thus by \cref{t:positive}, the weakest safe strategy for player $\pB$ can be effectively computed.
Our solver, described in \cref{s:implementation} finds it to be
\[
B(a, i, b) \coloneqq \neg E\lor\neg\Next(a,b)\lor\left( 
\begin{array}{l}
\forall n.\ (i \geq n \lor b \neq i)\ \land\\
\forall n.\ (\neg \Btw(b,i,n) \lor i > n)
\end{array}
\right)
\]
%
% \Btw (x,a,x) = \ffalse
% \Btw (x,b,b) = \ffalse
%
where $E$ axiomatizes the ring architecture, i.e., the predicate
$\Btw$ as the transitive closure of $\Next$ together with the predicate $\leq$.
The given strategy is weaker than the intuitive (and also safe) strategy of $(i = a)$,
and allows for more behaviours --- for example $a$ can send messages that are
greater than its own id in case they are not greater than the ids of nodes along the
way from $b$ back to $a$. 
% It also allows any $i$ in case that $b$ is not the next node
% in the ring, since messages are only sent if this is the case.
\qed
\end{example}

In general, though, existential SO quantifier elimination must be applied
to universally quantified formulas which cannot be brought into simple normal form.
In particular, we provide
a sequence of \emph{candidates} for the Hilbert choice operator which 
provides the \emph{weakest} Hilbert choice operator --- whenever it is FO definable.
\ignore{
The candidates in this sequence are then correct for 
larger and larger universes.
}
Consider a formula $\phi$ in normal form \eqref{def:normal}.
% , i.e., given by
% \[
% {\small
 % E\wedge
	% (\forall y.F\vee By)\wedge
	% (\forall y'.G\vee\neg By')\wedge
	% (\forall yy'.H\vee By\vee\neg By')}
% \]
%
Therein, the sub-formula $\neg H$ can be understood as a binary predicate between 
the variables $\bar y'$ and
$\bar y$ which may be composed, iterated, post-applied to predicates on $\bar y'$ and pre-applied to
predicates on $\bar y$.
We define $H^k$, $k\geq 0$, with free variables from $\bar y,\bar y'$ by
\[
{\small
\begin{array}{lll}
H^0 &=& \bar y\neq \bar y'	\\
H^k &=&	\forall \bar y_1.H^{k-1}[\bar y_1/\bar y']\vee H[\bar y_1/\bar y]\qquad\text{for}\;k>0
\end{array}}
\]
We remark that by this definition, 
{\small
\begin{eqnarray*}
H^{k+l} &=& \forall y_1.H^{k}[\bar y_1/\bar y']\vee H^l[\bar y_1/\bar y]
\end{eqnarray*}}
for all $k,l\geq 0$.
% This follows by associativity of composition of relations.
%
Furthermore, we define the formulas:
{\small
\begin{eqnarray*}
G\circ H^k	&=& \forall \bar y.G[\bar y/\bar y']\vee H^k	\\
G\circ H^k\circ F 
		&=& \forall\bar  y'.(G\circ H^k)\vee F[\bar y'/\bar y]
\end{eqnarray*}}
Then, we have:
\begin{lemma}\label{l:elim}
If $\exists B.\phi$ is FO definable, then it is equivalent to 
$E\wedge\bigwedge_{i=0}^k G\circ H^i\circ F$ for some $k\geq 0$.
\ignore{
\begin{enumerate}
\item	For all $k\geq 0$, $\exists B.\phi$ implies $E\wedge G\circ H^k\circ F$;
\item	If $\exists B.\phi$ implies some % universal 
	FO formula $\phi$, then for some $k\geq 0$, 
	$E\wedge\bigwedge_{i=0}^k G\circ H^i\circ F$ implies $\exists B.\phi$;
\item	If $\exists B.\phi$ is equivalent to some FO formula,
	then it is equivalent to 
	$E\wedge\bigwedge_{i=0}^k G\circ H^i\circ F$ for some $k\geq 0$.
	\qed
\end{enumerate}}
\end{lemma}

\noindent
Starting from $G$ and iteratively composing with $H$, provides us with a sequence of 
candidate SO Hilbert choice operators. Let
{\small
\begin{eqnarray}
\gamma_k	&=& \begin{array}{l}
		\neg E\vee\bigwedge_{i=0}^k (G\circ H^i)[\bar y/\bar y']
		\end{array}\label{def:gamma_k}
\end{eqnarray}}
for $k\geq 0$. The candidate $\gamma_k$ 
takes all $i$fold compositions of $H$ with $i\leq k$ into account.
Then the following holds:

\begin{lemma}\label{l:hilbert}\label{l:approx_hilbert}
For every $k\geq 0$,
\begin{enumerate}
\item	$\phi[\gamma_k/B]$ implies $\exists B.\phi$;
%
% NEW!!
%
\item	If $\exists B.\phi$ is equivalent to $\phi[\psi/B]$ for some FO formula $\psi$,
	then $\psi\implies\gamma_k$.
\ignore{
%
% MOVED AS SEPARATE LEMMA TO SECTION APPROX
%
\item	If $k> N^{r}$, then $\phi'[\gamma_k/B]$ is equivalent to $\exists B.\phi$
	in all universes up to size $N$;
}
\item	$\gamma_{k+1}\implies\gamma_{k}$, and
	if $\gamma_k\implies\gamma_{k+1}$, then 
	% $\bigwedge_{i=0}^k G\circ H^i\implies G\circ H^{k+1}$, then	\\
	$\phi[\gamma_k/B]=\exists B.\phi$.
\end{enumerate}
\end{lemma}

\noindent
As a result, the $\gamma_k$ form a decreasing sequence of candidate strategies
for safety player $\pB$.
%
% A nontrivial sufficient condition for ultimate stability is provided in 
% \cref{a:hilbert}.
%
We remark that due to statement (2), the sequence $\gamma_k$ results in the \emph{weakest}
Hilbert choice operator --- whenever it becomes stable.
\ignore{
 indeed implies that , though, that the condition 
$\bigwedge_{i=1}^k G\circ H^i =\bigwedge_{i=1}^{k+1} G\circ H^i$
% $\gamma_{k+1} = \gamma_k$ 
from statement (2) of 
lemma \ref{l:hilbert}, is \emph{sufficient} but not necessary for $\exists B.\phi'$ 
to be FO definable.
}

\noindent
We close this section by noting that there is a SO Hilbert choice operator 
which can be expressed in SO logic itself. The following theorem is related 
to Corollary 6.20 of \cite{gabbay2008second}, but avoids the explicit use
of fixpoint operators in the logic.

\begin{restatable}{theorem}{tsochoice}\label{t:SO_choice}
The weakest Hilbert choice operator $\mathcal{H}_B\phi$ for the universal formula \eqref{def:normal}
is definable by the SO formula:
\[
{\small
\neg E\lor\exists B. B\bar y\wedge (\forall\bar y'.G\vee\neg B\bar y')\wedge
		(\forall\bar y\bar y'.H\vee B\bar y\vee\neg B\bar y')
}
\]
\end{restatable}

\noindent
The weakest Hilbert choice operator itself can thus be obtained
by SO existential quantifier elimination.
The proof is by rewriting the formula and can be found in \cref{a:tsochoiceproof}.
\section{Implementation}\label{s:implementation}\label{s:experimental}

We have extended our solver NIWO for FO transition systems \cite{DBLP:conf/csfw/0008SZ18} 
to a solver for FO safety games which is able to verify inductive universal FO invariants 
and extract corresponding winning strategies for safety player $\pB$.
It has been packaged and published under \cite{helmut_seidl_2019_3514277}.
Our solver supports \emph{inference} of inductive invariants if the given candidate invariant 
is not yet inductive. 
For that it relies on the abstraction techniques from \cite{DBLP:conf/csfw/0008SZ18}
to strengthen arbitrary FO formulas by means of FO formulas using universal FO quantifiers
only.
For the simplification of FO formulas as well as for satisfiability of BSR formulas,
it relies on the EPR algorithms of the automated theorem prover Z3 \cite{de2008z3}.

We evaluate our solver on three kinds of benchmark problems.
First, we consider FO games with safety properties such as the running example
``Conference, Safety'' from \cref{fig:easychair,ex:release}.
For all of its variants, the fixpoint iteration \eqref{def:fixpoint} terminates in less than one second 
with a weakest winning strategy (w.r.t.\ the found inductive invariant)
whenever possible.
The second group ``Leader Election'' considers variants of the leader election protocol from 
\cref{e:leader_election}, initially taken from
\cite{padon2016ivy}.
Since the inductive invariant 
implies some transitively closed property, it cannot be inferred
automatically by our means. Yet, our solver succeeds in proving the
invariant from
\cite{padon2016ivy} inductive, and moreover, infers a FO definition
for the message to be forwarded to arrive at a single leader.
The third group ``Conference, NI'' deals with noninterference 
for variants of the conference management example
where the acyclic version has been obtained by unrolling the loop twice.
The difference between the \emph{stubborn} and \emph{causal} settings is
the considered angle of attack (see \cite{ATVA16} for an in-depth explanation).
In the setting of stubborn agents, the attackers try to break
the Noninterference property with
no specific intent of working together.
Here, the solver infers inductive invariants together 
with winning FO strategies (where possible) in $5-7$ seconds.
The setting of causal agents is inherently more complex as it
allows for groups of unbounded size that are working together 
to extract secrets from the system.
This allows for elaborate attacks where multiple
agents conspire to defeat noninterference \cite{DBLP:conf/ccs/Finkbeiner0SZ17}. 
The weakest strategy that is safe for stubborn agents 
($\neg \Conf(x,p)$ as a strategy for $B_1$) can can no longer be proven correct ---
instead the solver finds a counterexample for universes of size $\geq 5$.
To infer an inductive invariant and a safe strategy for causal agents, 
multiple iterates of the fixpoint iteration from \cref{def:fixpoint} must be
computed. Each iteration requires formulas to be brought into conjunctive normal form
--- possibly increasing formula size drastically.
To cope with that increase, 
formula simplification turns out not to be sufficient.
We try two different approaches to overcome this challenge: 
First, we provide the solver with parts of the inductive invariant, so fewer strengthening
steps are needed.
Given the initial direction, inference terminates much faster and provides us with a useful strategy.
For the second approach, we do not supply an initial invariant, but
accelerate fixpoint iteration by further strengthening of formulas. 
This enforces termination while still verifying safety. 
The extracted strategy, though, is much stronger 
and essentially rules out all intended behaviours 
of the system.

\begin{figure}[bt]
% \begin{adjustbox}{angle=90}
\centering
{\small
\begin{tabular}{| l | r | r || r | r | r | r |}
\hline
Name & Mode & Size & Invariant & \#Str. & Max. inv. & Time \\
\hline
\hline
Conference, Safety & synthesis & 6 & inferred & 4 & 50 & 736 ms\\
\hline
\hline
Leader Election & verification & 4 & inductive & 0 & 42 & 351 ms\\
\hline
Leader Election & synthesis & 4 & inductive & 0 & 42 & 346 ms\\
\hline
\hline
Conference, NI, stubborn & verification & 6 & inferred & 4 & 850 & 6782 ms\\
\hline
Conference, NI, stubborn & synthesis & 6 & inferred & 4 & 850 & 6817 ms\\
\hline
\hline
% strategy not false, no CNF approx
Conference acyclic, NI, causal & synthesis & 8 & inferred & 4 & 137 & 1985 ms\\
\hline
% changes with higher approx elim, strategy false
Conference, NI, causal & verification & 11 & counterex. & 7 & - & 2114 ms\\
\hline
Conference, NI, causal & synthesis & 11 & inferred & 2 & 102 & 2460 ms\\ 
\hline
Conference, NI, causal, approx. & synthesis & 11 & inferred & 8 & 5090 & 3359 ms\\
% counterexample
\hline
\end{tabular}}
% \end{adjustbox}
\caption{\label{fig:data}Experimental Results}
\end{figure}

All benchmarks were run on a workstation running Debian Linux on
an Intel i7-3820 clocked at $3.60$ GHz with $15.7$ GiB of RAM.
The results are summarized in \cref{fig:data}.
The table gives the group and type of experiment as well as the size of
the transition system in the number of nodes of the graph.
For the examples that regard Noninterference, the
agent model is given.
For verification benchmarks, the solver either proves the given invariant inductive
or infers an inductive invariant if the property is not yet inductive.
For synthesis benchmarks, it additionally extracts a universal formula for each $B \in \relsB$ 
to be used as a strategy.
The remaining columns give the results of the solver:
Could the given invariant be proven inductive, could an inductive strengthening be found, 
or did the solver find a counterexample violating the invariant?
We list the number of times any label of the invariant needed to be 
strengthened during the inference algorithm, the size of the largest label formula 
of the inferred invariant measured in the number of nodes of the syntax tree
and the time the solver needed in milliseconds (averaged over $10$ runs). 

Altogether, the experiments confirm that verification of provided invariants
as well as synthesis of inductive invariants and winning strategies
is possible for nontrivial transition systems with safety as well as noninterference objectives.

\section{Related Work}\label{s:refs}

In AI, First Order Logic has a long tradition for representing
potentially changing states of the environment \cite{brachman1992knowledge}.
First-order transition systems have then been used
to model reachability problems
that arise in robot planning
(see, e.g., chapters 8-10 in \cite{russell2016artificial}). 
The system GOLOG~\cite{golog}, for instance, is a
programming language based on 
FO logic.
A GOLOG program specifies the behavior of the agents in the system.
The program is then evaluated with a theorem prover, and thus
assertions about the program can be checked for validity.
Automated synthesis of predicates to enforce safety of the resulting system 
has not yet been considered.

There is a rich body of work on \emph{abstract state machines} (ASMs)~\cite{gurevich2018evolving}, i.e., state machines whose states are first-order structures. ASMs have been used to give comprehensive specifications of programming languages such as Prolog, C, and Java, and design languages, like UML and SDL (cf.~\cite{Boerger2003history}). A number of tools for the verification and validation of ASMs are available~\cite{Boerger2003tools}. Known decidability results for ASMs require, however, on strong restrictions such as sequential nullary ASMs~\cite{DBLP:phd/dnb/Spielmann00}.

In \cite{padon2015decentralizing,ball2014vericon}, it is shown that the
semantics of switch controllers of software-defined networks 
as expressed by \emph{Core SDN} can be nicely compiled into 
FO transition systems. The goal then is to use this translation 
to verify given invariants by proving them inductive.
Inductivity of invariants is checked by means of the theorem prover Z3 \cite{de2008z3}.
The authors report that, if their invariants are not already inductive,
a single strengthening, corresponding to the computation of $\Psi^{(1)}$
is often sufficient.
In \cite{Sagiv16-inferring}, the difficulty of inferring 
universal inductive invariants is investigated
for classes of transition systems whose
transition relation is expressed by FO logic formulas 
over theories of data structures.
The authors show that inferring universal inductive 
invariants is decidable when the transition relation is expressed by
formulas with unary predicates and a single binary predicate
restricted by the theory of linked lists
and becomes undecidable as soon as the binary symbol is not restricted
by background theory.
By excluding the binary predicate, this result is related to our 
result for transition systems with monadic predicates, equality and disequality, but
neither $\pA$- nor $\pB$-predicates. 
In \cite{karbyshev2017property}, an inference method is provided for universal invariants
as an extension of Bradley’s PDR/IC3 algorithm for inference of propositional invariants
\cite{bradley2011sat}.
The method is applied to variants of FO transition systems (no games) within a fragment of 
FO logic which enjoys the finite model property and is decidable.
Whenever it terminates, it either returns a universal invariant which
is inductive, or a counter-example.
% Ivy
This line of research has led to the tool
\textsc{Ivy} which generally applies FO predicate logic for the verification of
parametric systems \cite{padon2016ivy,mcmillan2018deductive}.
Relying on a language similar to \cite{padon2015decentralizing,ball2014vericon}, 
it meanwhile has been used, e.g., for the verification of network protocols such as 
leader election in ring topologies and the PAXOS protocol \cite{padon2017paxos}. 

% Game theory: Reachability Games?
In \cite{DBLP:conf/ccs/Finkbeiner0SZ17,ATVA16,DBLP:conf/csfw/0008SZ18}, 
hypersafety properties such as 
noninterference are studied for \emph{multi-agent workflows}.
These workflows are naturally generalized by our notion of FO transition systems.
The transformation in \cref{a:ni} for reducing noninterference to 
universal invariants originates from \cite{DBLP:conf/csfw/0008SZ18} where also
an approximative approach for inferring inductive invariants is provided. 
When the attempt fails, a counter-example can be extracted ---
but might be spurious.

All works discussed so far are concerned with verification rather than synthesis.
For synthesizing controllers for systems with an infinite state space, several approaches have been introduced
that automatically construct, from a symbolic description of a given concrete game, a finite-state abstract game~\cite{cegarcontrol,dealfaro07,DBLP:conf/fsttcs/DimitrovaF08,10.1007/978-3-642-33365-1_9,6987617}.
The main method to obtain the abstract state space is predicate abstraction, which
partitions the states according to the truth values of a set of predicates. States
that satisfy the same predicates are indistinguishable in the abstract game.
The abstraction is iteratively refined by introducing new predicates.
Applications include the control of real-time systems~\cite{10.1007/978-3-642-33365-1_9} and
the synthesis of drivers for I/O devices~\cite{6987617}.
In comparison, our approach provides a general modelling framework
of First Order Safety Games to unify different applications
of synthesis for infinite-state systems.

\section{Conclusion}\label{s:conclusion}

We have introduced \emph{First Order Safety Games} 
as a model for reasoning about games on parametric systems
where attained states are modeled as FO structures.
We showed that this approach allows to model
interesting real-world synthesis problems from the domains network protocols and 
informationflow in multiagent systems.
We examined the important case where all occurring predicates are monadic or nullary
and provided a complete classification into decidable and undecidable cases.
For the non-monadic case, we concentrated on \emph{universal} FO safety properties.
We provided techniques for certifying 
safety and also designed methods for synthesizing FO definitions
of predicates as strategies to enforce the given safety objective.
We have implemented our approach and succeeded to infer contents of
particular messages in 
the leader election protocol from \cite{padon2016ivy} 
in order to prove the given invariant inductive.
Our implementation also allowed us to synthesize predicates 
for parametric workflow systems as in \cite{DBLP:conf/csfw/0008SZ18},
to enforce noninterference in presence of declassification.
In this application, however, we additionally must take into account that the synthesized formulas
only depend on predicates whose values are independent of the secret. 
Restricting the subset of predicates possibly used by strategies, turns FO safety games into 
\emph{partial information} safety games.
It remains for future work, to explore this connection in greater detail in order, e.g., 
to determine whether strategies can be automatically
synthesized which only refer to specific \emph{admissible} predicates and, perhaps,
also take the \emph{history} of plays into account.

% \printbibliography % BibLatex version
% natbib version
% LICS Format
% \bibliographystyle{ACM-Reference-Format}
% LNCS format
\bibliographystyle{splncs04}
\bibliography{refs}

\appendix

\section{Proof of Theorem \ref{tapprox}}\label{a:games}

\label{tapproxproof}

\tapprox*

\begin{proof}

First, we extend the \emph{weakest precondition} operator to paths $\pi$ in the control-flow graph
of the FO safety game $\T$.
Assume that $\pi$ ends in program point $v$, and $\Psi$ is a property for that point. 
The weakest precondition
$\sem{\pi}^\top \Psi$ of $\Psi$ is defined by induction on the length of $\pi$.
If $\pi=\epsilon$, then $\sem{\pi}^\top \Psi = \Psi$. 

Otherwise, $\pi=e\pi'$ for some edge $e=(v_1,\theta,v')$. Then
\[
{\small\sem{\pi}^\top \Psi = \sem{e}^\top(\sem{\pi'}^\top\Psi)}
\]
For any path through a FO safety game $\T$ chosen by player $\pA$, 
the weakest precondition can be used to decide if player $\pB$ can
enforce a given assertion.
\begin{lemma}\label{l:path}
Let $\pi$ be a path 
% with start point $v'$, 
and $\Init$ some initial condition and $\Psi$ an
assertion about the endpoint of $\pi$. 
Safety player $\pB$ has a winning strategy for all plays on $\pi$ iff 
$\Init\implies\sem{\pi}^\top \Psi$.
\end{lemma}

\begin{proof}
We proceed by induction on the length of $\pi$. 
If $\pi=\epsilon$, then safety player $\pB$ wins all games in universes $U$ with valuations $\rho$ 
starting in states $s$ with $s,\rho\models\Init$ iff
$\Init\implies\Psi$, and the assertion holds.
Now assume that $\pi= e\pi'$ where $e=(u,\theta,v)$.
Let $\Psi' = \sem{\pi'}^\top\Psi$. 
By inductive hypothesis, safety player $\pB$ has a winning strategy for
$\pi'$ with initial condition $\Psi'$. This means
that she can force to arrive at the end point in some $s$ such that
$s,\rho\models\Psi$, given that she can start in some $s'$
with $s',\rho\models\Psi'$. 
By case distinction on whether edge $e$ is an $\pA$- or a $\pB$-edge, 
this holds whenever
the play starts in some $s$ with $s,\rho\models\Init$.

For the reverse direction, assume that for every universe $U$ and valuation $\rho$ 
chosen by reachability player $\pA$,
safety player $\pB$ can force to arrive at the end point of $\pi$ 
by means of the strategy $\sigma_{U,\rho}$
in a state $s$ such that $s,\rho\models\Psi$ whenever 
the play starts in some state $s_0$ with $s_0,\rho\models\Init$.
Assume that $s_0$ is an initial state with $s_0,\rho\models\Init$.
Again, we perform a case distinction on the first edge $e$.
First assume that $e$ is a $\pB$-edge. Let $B$ denote the relation selected by strategy 
$\sigma_{U,\rho}$ for $e$. 
We construct the successor state $s_1$ corresponding to edge $e$ and relation $B$.
Since safety player $\pB$ can win the game on $\pi'$ when starting in $s_1$, we conclude by
inductive hypothesis that $s_1,\rho\models\Psi'$.
This means that $s_0\oplus\{B_e\mapsto B_1\},\rho\models\Psi'\theta$ and therefore 
$s_0,\rho\models\exists B_e.\Psi'\theta$.
If $e$ is an $\pA$-edge, then for every choice $A$ of reachability
player $\pA$, we obtain a successor state $s_1$ such that by inductive hypothesis,
$s_1,\rho\models\Psi'$. This means that for all $A$,
$s_0\oplus\{A_e\mapsto A\},\rho\models\Psi'\theta$ and therefore also
$s_0,\rho\models\forall A_e.\Psi'\theta$.
In both cases, $s_0,\rho\models\sem{e}^\top(\sem{\pi'}^\top\Psi)$ and the claim follows.
\qed
\end{proof}

\noindent
For any node in the graph of the game, we successively construct the conjunction of 
the weakest preconditions of longer and longer paths starting in this particular node.
By induction, we verify that
\[
{\small\Psi^{(h)}[v]	= \bigwedge\{\sem{\pi}^\top I\mid
	          \pi\,\text{path starting at}\,v,|\pi|\leq h\}}
\]
for all $h\geq 0$.
Thus, safety player wins on all games of length at most $h$ starting at $\start$
iff $\Init\implies\Psi^{(h)}[\start]$ holds, and 
%
% From that, 
the assertion of the lemma follows.
\qed
\end{proof}

\section{Games for Noninterference} \label{s:noninterference}\label{a:ni}

In a multi-agent application such as the conference management system in \cref{fig:easychair}, 
not only safety properties
but also 
% \emph{hypersafety} properties such as 
\emph{noninterference} properties are of interest
\cite{DBLP:conf/ccs/Finkbeiner0SZ17,ATVA16,DBLP:conf/csfw/0008SZ18}. 
Assume, e.g., that 
% for the conference management system in \cref{fig:easychair},
no PC member should learn anything about the
reports provided for papers for which she has declared conflict of interest.
Our goal is to devise a predicate $\Assign$ for the edge between program points
1 and 2 which enforces this property.
%
% In the following, we make these notions precise.

In order to formalize noninterference for FO transition systems, 
we require a notion of \emph{observations} of participating agents.
% In such a multi-agent system, not every agent has access to every piece of information.
Following the conventions in \cite{DBLP:conf/ccs/Finkbeiner0SZ17,ATVA16,DBLP:conf/csfw/0008SZ18}, 
we assume that from every predicate $R$ of rank at least 1, 
agent $a$ \emph{observes} the set of all tuples $\bar z$ so that $Ra\bar z$ holds.
Moreover, we assume that there is a set $\Omega$ of input predicates predicates whose 
values are meant to be \emph{disclosed} only to privileged agents.
In the example from \cref{fig:easychair}, we are interested in
a particular agent $a$. The predicate $A_2$ which provides reports for papers, constitutes a
predicate whose tuples are only disclosed to agent $a$
if they speak about papers with which $a$ has no conflict.

In general, we assume that for each input predicate $O\in\Omega$, 
we are given a FO declassification condition $\Delta_{O,a}$ which specifies the set 
of tuples $\bar y$ where the value of $O\bar y$ may be disclosed to $a$. 
In the example from \cref{fig:easychair}, 
we have:
\begin{equation*}
{
\small
\Delta_{A_2,a} \;=\;\neg\Conf(a,y_1)
}
\end{equation*}
Noninterference for agent $a$ at a program point $v$ then means
that for each predicate $R\in\relsS$, the set of tuples observable by $a$
does not contain information about the sets of nondisclosed tuples of the input predicates
in $\Omega$, i.e., stay the same when these sets are modified.
Hereby, we assume that the control-flow is public and does not depend on secret information.
In a conference management system, e.g., the control-flow does not depend on
on the contents of specific reviews or posted opinions on papers.

While noninterference is best expressed as a \emph{hypersafety} property $\phi_a$ 
\cite{DBLP:conf/ccs/Finkbeiner0SZ17,ATVA16},
we will not introduce that logic here, but remark 
that the verification of $\phi_a$ for a FO transition system $\T$
can be reduced to the verification of an ordinary safety property $\phi^2_a$ 
of the (appropriately defined) \emph{self-composition} of $\T$
\cite{DBLP:conf/csfw/0008SZ18}.

In the following we recall that construction for the case 
that all agents of the system $\T$ are \emph{stubborn} \cite{ATVA16}. 
Intuitively, this property means that all agents' decisions
are \emph{independent} of their respective knowledge about input predicates.
A more elaborate construction $\T_a^{(c)}$ works for \emph{causal} agents whose
decisions may take their acquired knowledge into account (see
\cite{DBLP:conf/csfw/0008SZ18} for the details).
%
%
%\itodo{
%	\begin{itemize}
%		\item Describe Noninterference with Declass
%		\item Describe stubborn agents?
%		\item Solve by Selfcomposition
%		\item Show example selfcomposition
%	\end{itemize}
%}
% \itodo{
% Introduce the notion of stubborn as well as for causal agents!}
% noninterference with declassification
% for FO transition systems 
% can be effectively reduced to ordinary safety of suitably defined
% FO transition system \cite{DBLP:conf/csfw/0008SZ18}. 
%
\emph{Self-composition} $\T_a^{(s)}$ of the FO transition system $\T$ for stubborn agents
keeps track of two traces of $\T$. For that,
a copy $R'$ is introduced for each predicate $R$ in $\relsS\cup\Omega$.
For a formula $\phi$, let $[\phi]'$ denote the formula obtained from
$\phi$ by replacing each occurrence of a predicate $R\in\relsS\cup\Omega$ with
the corresponding primed version $R'$.
Initially, predicates and their primed versions have identical values, but later-on may diverge
due to differences in predicates from $\Omega$.
The initial condition $\Init^2$ is obtained from the initial condition $\Init$
of $\T$ by setting
\begin{equation}
{\small\begin{array}{lll}\Init^2 &=& \Init\wedge\bigwedge_{R\in\relsS}\forall \bar z. R\bar z\leftrightarrow R'\bar z\end{array}}
\label{eq:Init2}
\end{equation}
where we assume that the length of the sequence of variables $\bar z$
matches the rank of the corresponding predicate $R$.
The first track of $\T_{a}^{(s)}$ operates on the original predicates from $\relsS$, 
while the second track executes the same operations, but on the primed predicates.
%
% Subsequently, we only sketch the construction for the case of \emph{stubborn} agents,
% i.e., where the decisions of agents are not influenced by the secret.
%
Let $(u,\theta,v)$ denote a transition of the original system $\T$.
First assume that $\theta$ does not query any of the predicates from $\Omega$.
Then the self-composed system has a transition $(u,\theta^2,v)$ where
$\theta^2$ agrees with $\theta$ on all predicates from $\relsS$
where the right-hand side of $\theta^2$ for $R'$ is obtained from the right-hand side
of $\theta$ for $R$ by replacing each predicate $R\in\relsS$ with its primed counterpart.

When all agents are assumed to be \emph{stubborn}, their choices may not
depend on their acquired knowledge about the input predicates from $\Omega$.
For that reason, the same predicates from $\relsA$ are used on both tracks 
of the self-composition.
Thus, e.g., the second substitution applied in the loop of 
\cref{fig:easychair}
is extended with
\[
{\small \Review'(x,p,r) \wadd \Assign'(x,p)\land A_3(x,p,r)}
\]
% \itodo{Example from the workflow!}
Next, assume that $\theta$ has no occurrences of predicates in $\pA\cup\pB$,
but accesses some predicate $O\in\Omega$.
% Assume further that we are interested in noninterference for the particular agent $a$.
%
Then the value of the declassification predicate $\Delta_{O,a}$ for agent $a$ is 
queried to determine for which arguments $O$ is allowed to \emph{differ}
on the two tracks.
% receive identical results for agents where 
% for whom declassification returns $\true$, and possibly different
% results for the others. 
%
Accordingly, the substitution $\theta$ in $\T$ is replaced with $\theta^2_a$, followed by
{\small
\begin{eqnarray}
% \theta^2_{a} &=&\theta^2;\theta_{O,a} \qquad\text{where}\nonumber	\\
% \theta_{O,a} &=&\{\begin{array}[t]{@{}l}
\!\!	O\bar y	&\coloneqq& Ay \nonumber \\
\!\!	O'\bar y&\coloneqq& % \begin{array}[t]{@{}l}
			\Delta_{O,a}\land[\Delta_{O,a}]'\land A\bar y\lor % \\
			(\neg\Delta_{O,a}\lor\neg[\Delta_{O,a}]')\land A'\bar y
			% \end{array}
	% \end{array}
	\label{eq:oracle}
\end{eqnarray}}
for fresh input predicates $A, A'$ controlled by player $\pA$.
% where ``;'' means that 
% first the left argument substitution is applied and then the right argument substitution.
%
In \cref{fig:selfcomposed}, this means that at the edge from program point 2 to 3, we
have
\[
{\small
\!\!\!\begin{array}{lll}
\Review (x,p,r)\! &%\coloneqq
		{:=}& \!\Assign (x,p) \land A_2(x,p,r)	\\
\Review'(x,p,r)\! &%\coloneqq
		{:=}& \!\Assign'(x,p) \land \\
        &&     \begin{array}[t]{@{}l}
              (\neg \Conf(a,p)\land\neg\Conf'(a,p)\land A_2(x,p,r)\ \lor \\
              (\Conf(a,p)\lor\Conf'(a,p))\land A_2'(x,p,r))
                               \end{array} 	
\end{array}}
\]
Thus, classification on one of the tracks suffices to allow distinct results
when querying the oracle.

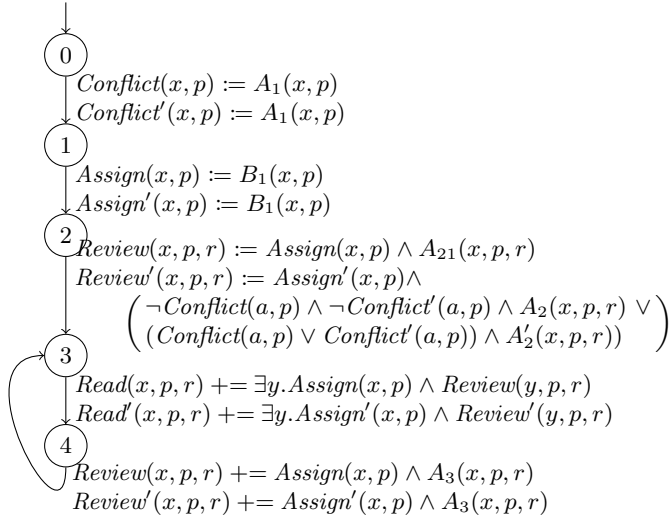
\begin{figure}[h!]
	\begin{tikzpicture}[
  node distance = 1.2cm
 ]
\useasboundingbox (-1,0.6) rectangle (9.2,-6.2);
	\node[fotsnode, initial above, initial text={}] (0) at (0,0) {0};
	\node[fotsnode] (1) [below of=0] {1};
	\node[fotsnode] (2) [below of=1] {2};
	\node[fotsnode,yshift=-0.4cm] (3) [below of=2] {3};
	\node[fotsnode] (4) [below of=3] {4};

	\path[->, draw]
		(0) edge node[right, align=left] {
			$\Conf(x,p) \coloneqq A_1(x,p)$\\
			$\Conf'(x,p) \coloneqq A_1(x,p)$
			} (1)
		(1) edge node[right, align=left] {
			$\Assign(x,p) \coloneqq B_1(x,p)$\\
			$\Assign'(x,p) \coloneqq B_1(x,p)$
		} (2)
		(2) edge node[right, align=left] {
			$\Review(x,p,r) \coloneqq \Assign(x,p) \land A_{21}(x,p,r)$\\
			$\Review'(x,p,r) \coloneqq \Assign'(x,p) \land$ \\
			\qquad
			$ \left(
				\begin{array}{l}
				\neg \Conf(a,p)\land\neg\Conf'(a,p)\land A_2(x,p,r)\ \lor \\
				(\Conf(a,p)\lor\Conf'(a,p))\land A_2'(x,p,r))
				\end{array} \right)$
 			} (3)
		(3) edge node[right, align=left] {
			$\Read(x,p,r) \wadd \exists y. \Assign(x,p) \land \Review(y,p,r)$\\
			$\Read'(x,p,r) \wadd \exists y. \Assign'(x,p) \land \Review'(y,p,r)$
		} (4);

		\path[->, draw] (4) edge [out=-90,in=180, looseness=2]
			node[pos=0.1, yshift=-0.6em, right, align=left] {
				$\Review(x,p,r) \wadd \Assign(x,p) \land A_3(x,p,r)$\\
				$\Review'(x,p,r) \wadd \Assign'(x,p) \land A_3(x,p,r)$\\
			}  (3);
\end{tikzpicture}	
	\caption{\label{fig:selfcomposed}Self-composition of the FO transition system from 
	\cref{fig:easychair}}
\end{figure}

The safety property $\phi_a^2$ to be verified for $\T_{a}^{(s)}$ then amounts to:
\begin{equation}
{\small\begin{array}{l}\bigwedge_{R\in\relsS}\forall \bar z.Ra\bar z\leftrightarrow R'a\bar z\end{array}}
\label{eq:ni}
\end{equation}
where we assume that the length of the sequence of variables $a\bar z$
matches the rank of the corresponding predicate $R$.
Let $\T_{a}^{(s)}$ denote the FO transition system obtained from $\T$ for stubborn agents 
in this way.
In case that player $\pB$ has no choice,
an adaptation of Theorem 1 from \cite{DBLP:conf/csfw/0008SZ18} implies that 
that $\T_{a}^{(s)}$ satisfies the invariant \eqref{eq:ni} for all program points $u$
iff $\phi_a$ % NI with declassification 
holds for $\T$.
This correspondence can be extended to an FO safety game $\T$ where the set of $\pB$ predicates
$\relsB$ is non-empty. However, we must ensure that the FO formulas describing the winning strategy $\sigma$ 
of the self-composition can be translated back into a meaningful strategy for $\T$.
A meaningful sufficient condition is that for each $B\in\relsB$, $B\sigma$ depends only
on predicates $R\in\relsS$ for which $R$ and $R'$ are equivalent.
More generally, assume that we are given for each program point $u$, a set $\rels_u$
where
\begin{equation}
{\small\begin{array}{l}\forall\bar y.R\bar y\leftrightarrow R'\bar y\qquad(R\in\rels_u)\end{array}}	\label{eq:NI}
\end{equation}
holds whenever program point $u$ is reached.
Then the strategy $\sigma$ for $\T_a^{(s)}$ is \emph{admissible} if for
each edge $(u,\theta,v)$ containing some $B\in\relsB$,
$B\sigma$ contains predicates from $R\in\rels_u$ only.
Due to \eqref{eq:NI}, the formulas $B\sigma$ and $[B\sigma]'$ then are equivalent.
Therefore, we obtain:

\begin{theorem}\label{t:partial}
Let $\T$ be an FO safety game with initial condition $I$ and subset $\rels_u\subseteq\relsS$ 
of predicates for each program point $u$ of $\T$.
Assume that $\sigma$ is a strategy so that for each 
predicate $B$ occurring at some edge $(u,\theta,v)$, the FO formula $B\sigma$ 
only uses predicates from $\rels_u$. 
Let $\T_{a}^{(s)}$ denote the corresponding
FO game with respect to stubborn agents and declassification predicates $\Delta_{O,a}$,
and assume that for each program point $u$, property \eqref{eq:NI} holds whenever $u$
is reached by $\T_a^{(s)}\sigma$.
Then the following two statements are equivalent:
\begin{enumerate}
\item
$\T\,\sigma$ satisfies the noninterference property $\phi_a$;
\item
$\T_a^{(s)}\sigma$ satisfies the safety property $\phi_a^2$.
\end{enumerate}
In particular, each admissible winning strategy for the FO safety game $\T_{a}^{(s)}$ gives rise to
a strategy for $\T$ that enforces noninterference.
\end{theorem}

Finding a strategy that enforces noninterference, thus turns into the synthesis problem for an
FO safety game --- with the extra obligation that 
potential winning strategies only
access subsets of \emph{admissible} predicates only.
In the conference management workflow from \cref{fig:easychair} with stubborn agents,
a strategy is required at the edge from program point $1$ to program point $2$. 
At that point, no secret has yet been encountered. Therefore, \emph{all} predicates are 
admissible ---
implying that \emph{any} winning strategy for the FO safety game in
\cref{fig:selfcomposed} can be translated back to a strategy which enforces
noninterference in $\T$.
In particular, we obtain (via $\T_{a}^{(s)}$) that any FO formula $\psi_a$ 
guarantees noninterference for which $\psi_a\implies\neg\Conf(y_1,y_2)$ holds.

% \input{content/partial}
%%%
\section{Proof of Theorem~\ref{t:counter}}
\label{a:counter}

\thmcounterex*

We first consider the case where there are no equalities, but $\pA$- as well as $\pB$-edges.
We show how to construct a safety game $G$ such that an automaton $M$ 
with multiple counters has a 
run, starting with empty counters and reaching some designated state $\term$ iff
safety player $\pB$ has no winning strategy in $G$.
The states $q\in\{1,\ldots,n\}$ of $M$ are encoded into flags 
$f_1,\ldots,f_n$ where $f_1$ and $f_n=\term$
correspond to the initial and final states, respectively.
The invariant $I$ is given by $\neg\term$.
Each counter $c_i$ of $M$ is represented by a monadic predicate $P_i$.
Incrementing the counter means to add exactly one element to $P_i$.
In order to do so, we set all flags $f_i$ to $\tfalse$ 
whenever the simulation was faulty.
Accordingly, we use as initial condition
\[
{\small f_1\wedge \begin{array}{l}\bigwedge_{j>1}\neg f_j\wedge\bigwedge_{i} \forall x.\neg P_ix
		\end{array}}
\]
Consider a step of $M$ which changes state $f_l$ to $f_{l'}$ and increments
counter $c_i$.
The simulation is split into two steps, one $A$-step followed by one $B$-step.
The $A$-step uses the substitution:
\[
{\small
\theta_1= \{ 
	\begin{array}[t]{lll}
	P_iy	&\mapsto& P_iy\vee Ay,	\\
	f_{l''} &\mapsto& \left\{
			\begin{array}{ll}
		f_l\wedge(\exists x.Ax\wedge\neg P_ix)
			&\text{if}\;l'' = l'	\\
			\tfalse	&\text{if}\;l''\neq l',	
			\end{array}\right.	\\
	P'y	&\mapsto& Py\;\}
	\end{array}
}
\]
where $P'$ is meant to record the values of the predicate $P_i$ before the transition. 
By this transition, some flag $f_{l''}$ is set only when the new predicate $P_i$ has
received some new element.
By the subsequent second transition, safety player $\pB$ can achieve $\bigwedge_{l}\neg f_l$ 
whenever the predicate $A$ chosen by reachability player $\pA$ in the previous step,
has more than one element outside $P_i$:
\[
{\small
\theta_2= \{ 
	\begin{array}[t]{lll}
	P_{i'}y	&\mapsto& P_{i'}y,	\\
	f_{l''}	&\mapsto& f_{l''}\wedge 
		\forall x_1x_2. \left(
					\begin{aligned}
						&P_ix_1\vee \neg P'x_1\vee \\
						&P_ix_2\vee\neg P'x_2\vee \\
						&B x_1\vee\neg B x_2
					\end{aligned}
					 \right),\\
	P'y	&\mapsto& \tfalse \;\}
	\end{array}
}
\]
Decrement by 1 can be simulated analogously. 
Since counters can also be checked for 0, we find that safety player $\pB$ wins
a play iff either the simulation of the counters was erroneous or 
reachability player $\pA$ is not able to reach $\term$.
Accordingly, statement (1) of the theorem follows.

In the given simulation,
the $B$-predicates can be replaced by means of an equality in the substitution:
\[
{\small f_{l''} \mapsto f_{l''}\wedge 
		\forall x_1x_2. P_ix_1\vee \neg P'x_1\vee 
			 P_ix_2\vee\neg P'x_2\vee x_1= x_2}
\]
Therefore, also statement (2) follows.
A disequality would have served the same purpose if 
deviation from the correct simulation would have been tracked by means of an error flag.
This kind of simulation is exemplified for the proof of statement (3).

For statement (3), we introduce a dedicated error flag $\err$ and
sharpen the invariant to
\[
{\small \neg\err\wedge\begin{array}{l}
		(\bigvee_{j=1}^{n-1}f_j\vee\bigwedge_{j=1}^n\neg f_j)
		\end{array}}
\]
The error flag is initially assumed to be $\tfalse$, and used to force safety player $\pB$ 
to choose sets $B$ with appropriate properties. Thus, we use
\[
{\small\neg\err\wedge f_1\wedge \begin{array}{l}
	\bigwedge_{j>1}\neg f_j\wedge\bigwedge_{i} \forall x.\neg P_ix
	\end{array}}
\]
as initial condition.
For the actual simulation, we use a single program point together 
with edges for each transition of the counter machine.
Incrementing counter $c_i$ (combined with state transition from $q_l$ to $q_{l'}$), e.g.,
is simulated by an edge with the substitution
\[
\theta'= \{ 
	\begin{array}[t]{lll}
	P_iy	&\mapsto& P_iy\vee By,	\\
	f_{l''} &\mapsto& \left\{\begin{array}{ll}
		f_l	&\text{if}\;l'' = l'	\\
		\tfalse	&\text{if}\;l''\neq l'\;,
		\end{array}\right. \\
	\err	&\mapsto& \err \vee (\forall x.\neg Bx\vee P_ix)\;\vee	\\
		&&	(\exists x_1x_2.\begin{array}[t]{l}
			\neg(Bx_1\wedge\neg P_ix_1)\;\vee	\\
			\neg(Bx_2\wedge\neg P_ix_2)\vee x_1\neq x_2)\;\}
			\end{array}\\
	\end{array}
\]
Due to $\neg\err$ in the invariant, safety player $\pB$ is forced to choose a set $B$ which
adds exactly one element to $P_i$,
while the subformula $\bigwedge_{j}\neg f_j$ forces
reachability player $\pA$ is choose edges according to the state transitions of the counter machine.
\qed

\section{Proof of Theorem~\ref{t:monadic_plain}}\label{a:monadic_plain}

\thmmonadicplain*

\noindent
For the proof of \cref{t:monadic_plain}, we rely on a technique similar to the 
\emph{Counting Quantifier Normal Form} (CQNF) as introduced by Behmann
in \cite{behmann1922beitrage} and picked up in \cite{wernhard2017heinrich}.
A \emph{counting quantifier} $\exists^{\geq n} x. \phi(x)$ expresses that at least $n$ individuals exist for which $\phi$ holds, i.e.
\[
{\small	\exists^{\geq n} x. \phi \equiv \exists x_1\ldots x_n. 
	\begin{array}{l}
	\bigwedge_{1\leq i \leq n} \phi[x_i/x] \land \bigwedge_{i < j \leq n} x_i \neq x_j
	\end{array}}
\]
The main theorem is:
A monadic FO formula $\phi$ is said to be in \emph{liberal counting quantifier normal form}
(liberal CQNF) iff $\phi$ 
	is a Boolean combination of \emph{basic formulas} of the form:
	\begin{itemize}
	\item	${\small\exists^{\geq n}x.\ \bigwedge\limits_{1\leq i \leq m} L_i(x)}$
	
	where $n \geq 1$, $m \geq 0$, and the $L_i(x)$ 
	are pairwise different and pairwise non-complementary 
	positive or negative literals with unary predicates applied to the individual 
	variable $x$, and dis-equalities $x\neq a$ for free variables $a$,
	\item nullary predicates $P$, 
	\item $P(x)$, where $P$ is a unary predicate and $x$ is a global variable,
	\item $x = x'$, where $x,x'$ are global variables.
	\end{itemize}
$\phi$ is in \emph{strict counting quantifier normal form} (strict CQNF) 
if it is in liberal CQNF and additionally does not have dis-equalities of bound FO variables with
free FO variables.
We remark that the notion of strict CQNF has been called just CQNF in \cite{wernhard2017heinrich}.
We have:
	
\begin{theorem}[CQNF for Monadic FO Formulas \cite{wernhard2017heinrich,behmann1922beitrage}]\label{t:counting-qnf} 
	From each monadic FO formula
	$\phi$ equivalent FO formulas $\phi_1,\phi_2$ can be constructed such that
	\begin{enumerate}
	\item $\phi_1$ is in liberal CQNF;
	\item $\phi_2$ is in strict CQNF;
	\item all FO variables and predicates in $\phi_1, \phi_2$ 
		also occur in $\phi$. 
	\qed
	\end{enumerate}
\end{theorem}
We remark that the construction of $\phi_1$ in liberal CQNF follows the same lines
as the construction of $\phi_2$ where only the step of eliminating dis-equalities between
bound variables and free variables is omitted.

The transformation into strict CQNF is illustrated by the following example.
\begin{example}
{\small
	\begin{align*}
		& \exists y.\ \exists x.\ px \land x \neq y & \equiv\\
		& \exists y.\ (\exists^{\geq 1} x.\ px) \land ((\exists^{\geq 2} x.\ px) \lor \neg py) & \equiv\\
		& (\exists^{\geq 1} x.\ px) \land ((\exists^{\geq 2} x.\ px) \lor \exists y. \neg py) & \equiv\\
		& (\exists^{\geq 1} x.\ px) \land ((\exists^{\geq 2} x.\ px) \lor \exists^{\geq 1} y. \neg py)
	\end{align*}}
\end{example}
\noindent

For a monadic FO formula $\phi$ in liberal or strict CQNF, the \emph{quantifier rank} $\rank(\phi)$
equals the maximal $k$ such that $\exists^{\geq k}$ occurs in $\phi$.
Likewise, for a substitution $\theta$ where all images of predicates are in strict CQNF, 
$\rank(\theta)$ equals the maximal rank of a formula
in the image of $\theta$.
For the rest of this subsection, we assume that for all substitutions, all formulas in their
images are in strict CQNF.
We now state our results for such substitutions on monadic FO formulas in CQNF.

\begin{lemma}\label{l:monadicsubstitutions}
	Given a monadic FO formula $\phi$ in liberal CQNF and a substitution $\theta$,
	the quantifier rank of $\phi\theta$ in liberal CQNF is at most the maximum of $\rank(\phi)$ and 
	$\rank(\theta)$.
\end{lemma}

\begin{proof}
	Since $\phi$ is in liberal CQNF and $R\theta$ $(R\in\relsS)$ are all in strict CQNF,
	none of them contain equalities between bound variables, and
	All quantifier scopes $\exists^{\geq k} x.$ contain 
	only literals that mention $x$.
	While in $\phi$ these scopes may contain inequalities $x \neq a$ for free variables $a$,
	this is not allowed in the $R\theta$. In particular, there are no dis-equalities between
	$y$ and a bound FO variable.
	Thus, we can write $R\theta$ in the form 
	$\bigvee_{j=1}^{l_R}\psi_{R,j}(y) \wedge \psi'_{R,j}\vee\psi_R''$ where each $\psi_{R,j}(y)$ is 
	a \emph{quantifierfree} boolean combination of literals applied to $y$,
	equalities or dis-equalities of $y$ with further free variables, 
	and all $\psi'_{R,j}$, $\psi_R''$ do not contain $y$.
	Applying $\theta$ to a literal $L(a)$, where $a$ is a free variable, does not introduce
	new nested quantifiers. 
	Now consider a quantified basic formula 
	\[
	{\small\exists^{\geq k} x. 
	\begin{array}{l}
	(\bigwedge_{i=1}^{l_1} L_i(x))\wedge
	(\bigwedge_{i=1}^{l_2} \neg L'_i(x))\wedge D(x)
	\end{array}}
	\]
	of $\phi$ where $D(x)$ is a conjunction of disequalities with free variables of $\phi$.
	Application of $\theta$  
	yields a formula which is a boolean combination of
	\begin{itemize}
	\item	basic formulas from $\theta$ without occurrences of $y$
		since these can be extracted out of the scope of any quantifier
		of $\phi$;
	\item	basic formulas from $\phi$ without occurrences of predicates;
	\item	basic formulas arising of the CQNF of a formula
		\[
		\exists^{\geq k} (\bigwedge_{j=1}^{m} 
			\psi_{R_j,i_j}[x/y])\wedge\neg\psi_{R,i}[x/y]\wedge D(x)
		\]
		for some predicates $R_j,R$ and indices $i_j,i$.
		By construction, each formula $\psi_{R_j,i_j}[x/y]$ as well as 
		formula $\neg\psi_{R,i}[x/y]$ is quantifierfree. 
		Therefore, it is equivalent to a boolean combination of
		basic formulas of rank at most $k$.
	\end{itemize}
	Altogether, the rank of $\phi\theta$ is thus bounded by
	the maximum of the ranks of $\phi$ and $\theta$.
\end{proof}

\begin{lemma}
	For any monadic FO formula $\phi$ in liberal CQNF and a sequence of substitutions $\theta_0,\ldots,\theta_n$, 
	in strict CQNF, it holds that
	\[
		\rank(\phi \theta_0 \ldots \theta_n) \leq \max(\rank(\phi), \rank(\theta_0), \ldots, \rank(\theta_n))
	\] 
\end{lemma}
\noindent The proof follows from the repeated application of \cref{l:monadicsubstitutions}.

Now that we proved the intermediate steps, we can prove the initial \cref{t:monadic_plain}.

\begin{proof}[Proof of \cref{t:monadic_plain}]
For all $h \geq 0$ and nodes $v$, $\Psi^{(h)}[v]$ is a conjunction of sequences of 
substitutions $\theta$ from $E$ applied to some FO formula $I[v']$.
Thus, $\rank(\Psi^{(h)}[v])$ is at most
\[
\max (\{ \rank(\theta) \mid (v,\theta,v') \in E \} \cup 
\{ \rank(I[v']) \mid v' \in V \})
\]
Let $r$ be this maximum.
For a given set of constants, there are only finitely many formulas of fixed quantifier rank 
(up to logical equivalence).
Thus, fixpoint computation as given in \cref{s:games} necessarily terminates.
According to the proof of \cref{tapprox}, 
a game $G$ is safe iff for all $h \geq 0$, $\Init \implies \Psi^{(h)}[v_0]$.
Therefore, \cref{t:monadic_plain} follows.
\end{proof}

\noindent
We remark that the given finite upper bound $r$ to the ranks of all formulas $\Psi^{(h)}[v]$
together with
the finite model property \cite{borger2001classical} implies that
reachability player $\pA$ can win iff $\pA$ can win in a  
a universe of size at most $r' 2^{\abs{\relsS}}$ where $r'$ is the maximum of $r$ and the rank of
$\Init$.
 
\section{Proof of Theorem \ref{t:mono_A}} \label{proof_mono_A}

\thmmonoa*

We have:

\begin{lemma}\label{l:equality}
Let $\phi$ be a FO formula with free variables from $\Const$ 
possibly containing equalities or disequalities between bound variables.
We construct a formula $\phi^\sharp$ with free variables from $\Const$ 
and neither positive nor negative equalities between bound variables
such that the following holds:
\begin{enumerate}
\item	$\phi^\sharp\implies\phi$;
\item	If $\psi\implies\phi$ holds for any other monadic formula $\psi$
	without (dis-)equalities between bound variables,
	then $\psi\implies\phi^\sharp$.
\item There exists some $d\geq 0$ such that for a model $s$ of multiplicity at least $d$ and a valuation $\rho$,
	$s,\rho\models\phi^\sharp$ iff $s,\rho\models\phi$.
\end{enumerate}
\end{lemma}

\noindent
If the assumptions of lemma \ref{l:equality} are met, $\phi^\sharp$ is called the 
\emph{weakest strengthening} of $\phi$ by formulas without equalities.
\begin{proof}
Assume that $\phi$ is in prenex normal form and that the quantifierfree part
$\phi'$ is in disjunctive normal form.
By transitivity of equality, we may assume that in each monomial $m$ of $\phi'$
for each occurring equality $x=y$ one of the following properties holds:
\begin{itemize}
\item	both $x,y$ are free variables; or
\item	$x$ is free and $y$ occurs in the quantifier prefix; or
\item	neither $x$ nor $y$ are free, $x$ is different from $y$ and
	the leftmost variable in the quantifier prefix
	which is transitively equal to $y$.
\end{itemize}
Next, $m$ is rewritten in such a way that additionally 
no variable $y$ on a right side of an equality is existentially quantified.
As a result, each remaining right side of an equality literal is 
either free in $\phi'$ (in which case the left side is also free) 
or universally quantified. 
Now consider any model $s$ such that $\mu(s)\geq d$ for some $d>0$ exceeding 
the number of free variables plus the length of the quantifier prefix of $\phi$. 
Then we verify for each universally quantified variable $y$
(by induction on the number of universally quantified 
variables occurring in a quantifier prefix $Qz$),
that  $s,\rho\models\forall y\,Qz.\phi'$ iff $s,\rho\models\forall y\,Qz.\phi''$
where $\phi''$ is obtained from $\phi'$ by replacing each occurrence of an equality 
$x=y$ with $x \sim_{\Const} y$, defined as $\bigvee_{c\in\Const}x=c\wedge y=c$.

Accordingly, we construct $\phi^\sharp$ from $\phi$ by replacing all 
equalities $x=y$ where $y$ is universally quantified with 
$x \sim_{\Const} y$.
Then $\phi^\sharp$ satisfies statements (1) and (3).
In order to prove statement (2), we first observe that $\psi\implies\phi$ also holds for
all models $s$ with $\mu(s)\geq d$ for all values of $d$ exceeding the cardinality of $\Const$. 
By property (3), we therefore have that
$\psi\implies\phi^\sharp$ in all models $s$ and all valuations $\rho$
where $\mu(s)$ is sufficiently large.
Since (dis-)equalities in $\phi^\sharp$ and in $\psi$ are not applied 
to pairs of bound variables,
% only refer to free variables,
the assertion follows.
\end{proof}

\noindent
We conclude:

\begin{corollary}\label{c:abstract_eq}
Assume that $\phi,\phi'$ are monadic FO formulas with positive occurrences of equality only.
Then
\begin{enumerate}
\item	$(\phi\wedge\phi')^\sharp = \phi^\sharp\wedge(\phi')^\sharp$, and
\item	$(\forall A.\phi)^\sharp = (\forall A.\phi^\sharp)^\sharp$
\end{enumerate}
\end{corollary}

\noindent
With this, we can now prove the initial \cref{t:mono-mc}.

\begin{proof}[Proof of \cref{t:mono-mc}]
Let $\Psi^{(h)}$ denote the $h$th iteration of the weakest precondition \eqref{def:fixpoint}
as defined in section \ref{s:games}.
Due to SO Quantifier Elimination as in \cite{behmann1922beitrage},
each formula $\Psi^{(h)}[v]$ is equivalent to a monadic FO formula.
If neither $I$ nor $\theta$ contain equalities, $\Psi^{(h)}[v]$ has positive occurrences of
equalities only.

\noindent
The sequence $\Psi^{(h)}[v]$ for $h\geq 0$ still need not stabilize
as more and more FO variables may be introduced.
Let $\Psi_0^{(h)}$ denote the $h$th iteration of the \emph{abstraction} of 
the weakest precondition:
\[
{\small
\begin{array}{lll}
\Psi_0^{(0)}[v]	&=&	I[v]	\\
\Psi_0^{(h)}[v]	&=&	\Psi_0^{(h-1)}\;\wedge\\
	&&\bigwedge\limits_{(v,\theta,v')\in E}
			(\forall A_e.(\Psi_0^{(h-1)}[v']\theta))^\sharp\quad
			\text{for}\; h>0
\end{array}}
\]
where the abstraction operator $(\cdot)^\sharp$ returns the weakest strengthening
by means of a monadic FO formula without equality.
Recall from corollary \ref{c:abstract_eq}
that the abstraction operator commutes with conjunctions.
Also, we have that $(\forall A_e.\phi)^\sharp = (\forall A_e.\phi^\sharp)^\sharp$
for each monadic FO formula with positive occurrences of equality only.
By induction on $h$, we find that
$\Psi_0^{(h)}[v] = (\Psi^{(h)}[v])^\sharp$ holds for all $h\geq 0$.
Since $\Init$ does not contain equalities, we therefore have for all $h\geq 0$, that 
$\Init\implies\Psi^{(h)}[\start]$ iff
$\Init\implies\Psi_0^{(h)}[\start]$.
Since (up to equivalence) the number of monadic formulas without
equalities or disequalities is finite, the sequence $\Psi_0^{(h)}$
for $h\geq 0$ eventually stabilizes.
This means that there is some $h'\geq 0$ such that for each program point $v$,
$\Psi_0^{(h')}[v] = \Psi_0^{(h'+1)}[v]$. Thus, game $G$ is safe iff
$\Init\implies\Psi_0^{(h')}[\start]$.
Since the implication is decidable, the theorem follows.
\end{proof}

\section{Proof of Theorem \ref{t:mono_B}}
\label{proof_mono_B}

\thmmonob*
The proof is analogous to the proof of theorem \ref{t:mono-mc}
where the abstraction of equalities now is replaced with an abstraction
of disequalities, and corollary \ref{c:abstract_eq} is replaced with a similar
corollary \ref{c:abstract_neq} dealing with disequalities.

In analogy to safety games with invariants containing equalities,
we provide a weakest strengthening of 
monadic FO formulas, now containing positive occurrences of disequalities only. 
Let $\phi$ denote a monadic formula in negation normal form
with free variables from $\Const$,
and no positive occurrences of equalities between bound variables.
We define $\phi^\sharp$ now as the formula obtained from 
$\phi$ by replacing each literal $x\neq y$ ($x,y$ bound variables) with 
\begin{equation}
{\small\begin{array}{l}
	(\bigvee_{R\in\rels}Rx\wedge\neg Ry\vee \neg Rx\wedge Ry)\;\vee\\
	(\bigvee_{c\in\Const}x=c\wedge y\neq c\vee x\neq c\wedge y=c)
	\end{array}}
\label{def:abstract_neq}
\end{equation}
Then, $\phi^\sharp\implies\phi$ holds, and we claim:

\begin{lemma}\label{l:abstract_neq}
Let $\psi$ be any monadic FO formula without equalities or disequalities between 
bound variables such that
$\psi\implies\phi$ holds. Then also $\psi\implies\phi^\sharp$ holds.
\end{lemma}

\begin{proof}
We proceed by induction on the structure of $\phi$.
Clearly, the assertion holds whenever $\phi$ does not contain disequalities
between bound variables.
Assume that $\phi$ is the literal $x\neq y$ for bound variables $x,y$.
Assume that $\psi\implies (x\neq y)$, but 
$\psi$ does not imply formula \eqref{def:abstract_neq}.
This means that there is a model $M$ and an assignment $\rho$ such that
$M,\rho\models\psi\wedge\bigwedge_{R\in\rels}Rx\wedge Ry\vee\neg Rx\wedge\neg Ry
	\wedge\bigwedge_{c\in\Const}(x\neq c\vee y=c)\wedge(x=c\vee y\neq c)$ holds.
W.l.o.g., $M$ is minimal, i.e., elements which cannot be distinguished by means of 
predicates in $\rels$ or free variables in $\Const$, are equal. 
But then $M,\rho\not\models (x\neq y)$ --- in contradiction
to the assumption.

Now assume that $\phi= \phi_1\wedge \phi_2$. 
Then $\phi^\sharp = \phi_1^\sharp\wedge\phi_2^\sharp$.
Let $\psi$ imply $\phi$. Then $\psi\implies\phi_i$ for each $i$.
Therefore, by induction hypothesis, $\psi\implies\psi_i^\sharp$ for all $i$.
As a consequence, $\psi\implies\phi^\sharp$.

Now assume that $\phi= \phi_1\vee\phi_2$. 
Then $\phi^\sharp = \phi_1^\sharp\vee\phi_2^\sharp$.
If $\psi$ implies $\phi$, then for each model $M$ variable assignment $\rho$,
there is some $i$ so that $M,\rho\models\psi\implies\phi_i$. 
Assume for a contradiction that $\psi\wedge\neg(\phi_1^\sharp\vee\phi_2^\sharp)$ is satisfiable.
Then there is some model $M$, assignment $\rho$ so that
$M,\rho\models\psi\wedge\neg\phi_1^\sharp\wedge\neg\phi_2^\sharp$.
In particular, there is some $i$ so that 
$M,\rho\models\phi_i\wedge\neg\phi_1^\sharp\wedge\neg\phi_2^\sharp$.
By inductive hypothesis, $\phi_i^\sharp\implies\phi_i$ holds. 
We conclude that
$M,\rho\models\phi_i\wedge\neg\phi_1\wedge\neg\phi_2$ holds --- contradiction.
Similar arguments also apply to existential and universal quantification in $\phi$.
As a consequence, $\psi\implies\phi^\sharp$ holds.
\end{proof}

\begin{corollary}\label{c:abstract_neq}
Assume that $\phi,\phi'$ are monadic FO formulas without positive occurrences of equalities
between bound variables.
Then
\begin{enumerate}
\item	$(\phi\wedge\phi')^\sharp = \phi^\sharp\wedge(\phi')^\sharp$, and
\item	$(\exists B.\phi)^\sharp = (\exists B.\phi^\sharp)^\sharp$
\end{enumerate}
\end{corollary}

\section{Proof of Lemma \ref{l:normal}}
\label{l:normalproof}

\lemnormalform*

\begin{proof}
W.l.o.g., we assume that $\phi=\forall x.\phi'$ where $\phi'$ is quantifierfree
and in conjunctive normal form. 
Let $E,F',G',H'$ equal the conjunction of all clauses in $\phi'$ containing
no occurrence of $B$, only positive, only negative and both positive and negative
occurrences of $B$.
Each clause of the form $c'\vee Bz_1\vee\ldots\vee Bz_k$ ($c'$ without $B$)
is equivalent to 
\[
{\small \forall y.c'\vee\begin{array}{l}
		(\bigwedge_{i=1}^kz_i\neq y)\vee By
		\end{array}}
\]
where $z_i\neq y$ abbreviates the disjunction $\bigvee_{j=1}^rz_{ij}\neq y_j$
--- given that $z_i=z_{i1}\ldots z_{ir}$.
Likewise, each clause of the form $c'\vee \neg Bz_1\vee\ldots\vee \neg Bz_k$ ($c'$ without $B$)
is equivalent to 
\[
{\small \forall y'.c'\vee\begin{array}{l}
		(\bigwedge_{i=1}^kz_i\neq y')\vee \neg By'
		\end{array}}
\]
Finally, each clause of the form 
$c'\vee \neg Bz_1\vee\ldots\vee \neg Bz_k\vee\neg z'_1\vee\ldots\neg Bz'_l$ ($c'$ without $B$)
is equivalent to 
\[
{\small
\forall yy'.c'\vee
\begin{array}{l}
(\bigwedge_{i=1}^kz_i\neq y)\vee 
(\bigwedge_{i=1}^lz'_i\neq y')
\end{array}
\vee By\vee \neg By'}
\]
Applying these equivalences to the clauses in the conjunctions in $F',G',H'$, respectively,
we arrive at conjunctions of clauses which all contain just the $B$-literal $By$, 
the $B$-literal $\neg By'$ or 
$By\vee\neg By'$, respectively. From these, the formulas $F,G$ and $H$ can be constructed
by distributivity.
% \qed
\end{proof}

\section{Proof of \cref{t:SO_choice}}
\label{a:tsochoiceproof}

\tsochoice*

\begin{proof}
Our goal is to prove that $\exists B.\phi$ implies the formula $\phi[{\cal H}_B\phi/B]$.
We consider each conjunct of $\phi$ in turn.
\[
{\small
    \begin{array}{lcl}
\forall \bar y. F\vee{\cal H}_B\phi &=&
\forall \bar y.\exists B. F\vee (B\bar y\;\wedge \\
&&(\forall \bar y'.G\vee\neg B\bar y')\wedge(\forall \bar y\bar y'.H\vee B\bar y\vee\neg B\bar y')) \\
&\leftarrow&
\forall \bar y.\exists B.  (F\vee B\bar y)\;\wedge\\
&&(\forall \bar y'.G\vee\neg B\bar y')\wedge(\forall \bar y\bar y'.H\vee B\bar y\vee\neg B\bar y')  \\
&\leftarrow&
\exists B.\forall \bar y.  (F\vee B\bar y)\;\wedge\\
&&(\forall \bar y'.G\vee\neg B\bar y')\wedge(\forall \bar y\bar y'.H\vee B\bar y\vee\neg B\bar y')  \\
&=&\exists B.\phi
\end{array}
}
\]
\[
{\small
    \begin{array}{lcl}
\forall \bar y'. G\vee\neg{\cal H}_B\phi &=&
    \forall \bar y'.\forall B.G\vee\neg B\bar y'\;\vee  \\
&&(\exists \bar y'.B\bar y'\wedge\neg G)\vee (\exists \bar y\bar y'.\neg H\wedge\neg B\bar y\wedge B\bar y')    \\
&=& \forall B.\forall \bar y'.G\vee\neg B\bar y'\;\vee  \\
&&(\exists \bar y'.\neg G\wedge B\bar y')\vee (\exists \bar y\bar y'.\neg H\wedge\neg B\bar y\wedge B\bar y')   \\
&=& \ttrue
    \end{array}
}
\]
\[
{\small
\begin{array}{cl}
\multicolumn{2}{l}{\forall \bar y\bar y'.H\vee{\cal H}_B\phi\vee\neg{\cal H}_B\phi[\bar y'/\bar y]} \\
=& 
\forall \bar y\bar y'.H\;\vee\\
&(\exists B. B\bar y\wedge
    (\forall \bar y'.G\vee\neg B\bar y')\wedge(\forall \bar y\bar y'.H\vee B\bar y\vee\neg B\bar y'))\;\vee \\
 &(\forall B.\neg B\bar y'\vee \neg(\forall \bar y'.G\vee\neg B\bar y')\vee
                \neg(\forall \bar y\bar y'.H\vee B\bar y\vee \neg B\bar y'))    \\
\leftarrow&
\forall \bar y\bar y'.\forall B.H\;\vee \\
&B\bar y\wedge(\forall \bar y'.G\vee\neg B\bar y')\wedge(\forall \bar y\bar y'.H\vee B\bar y\vee\neg B\bar y')\;\vee    \\
&\neg B\bar y'\vee \neg(\forall \bar y'.G\vee\neg B\bar y')\vee
                \neg(\forall \bar y\bar y'.H\vee B\bar y\vee \neg B\bar y') \\
= & 
\forall \bar y\bar y'.\forall B.H\vee B\bar y\vee\neg B\bar y'\;\vee    \\
 &\neg B\bar y'\vee \neg(\forall \bar y'.G\vee\neg B\bar y')\vee
                \neg(\forall \bar y\bar y'.H\vee B\bar y\vee \neg B\bar y') \\
=& \ttrue
\end{array}
}
\]
Altogether therefore,
$\exists B.\phi \implies \phi[{\cal H}_B\phi/B]$, and the assertion follows.  
\end{proof}

\end{document}